\documentclass[reqno]{tran-l}

\usepackage[utf8]{inputenc}
\usepackage{ifthen,mathrsfs, amsmath, amssymb, amsthm, braket, bm, graphicx, mathtools,enumitem,color,marginnote}
\usepackage[nocompress]{cite}
\usepackage{microtype}
\usepackage[english]{babel}
\usepackage[nodayofweek,long]{datetime}
\usepackage[colorlinks=true, allcolors=blue]{hyperref}
\hypersetup{final}
\usepackage[hang,flushmargin]{footmisc}
\usepackage{tikz-cd,upgreek}

\newcommand{\numorbitals}{N_\mathrm{b}}
\newcommand{\ctTBprefactor}{h_0}
\newcommand{\ctTBexponent}{{\gamma_0}}
\newcommand{\ctHamregularity}{\nu}
\newcommand{\ctnoninterpen}{\mathfrak{m}}
\newcommand{\ctCT}{\gamma_\mathrm{CT}}
\newcommand{\ctGamma}{{\Upsilon}}

\newcommand{\Ham}{\mathcal{H}}

\let\originalleft\left
\let\originalright\right
\renewcommand{\left}{\mathopen{}\mathclose\bgroup\originalleft}
\renewcommand{\right}{\aftergroup\egroup\originalright}

\renewcommand{\leq}{\leqslant}\renewcommand{\geq}{\geqslant}
\renewcommand{\above}[2]{\genfrac{}{}{0pt}{}{#1}{#2}}

\newcommand{\HamLin}{\Ham^\mathrm{L}}
\newcommand{\HamNonLin}{\Ham^\mathrm{NL}}

\newtheorem*{fact*}{Fact}
\newtheorem{theorem}{Theorem}[section]
\newtheorem*{theorem*}{Theorem}

\newtheorem{lemma}[theorem]{Lemma}

\theoremstyle{remark}
\newtheorem{remark}{Remark}
\numberwithin{equation}{section}

\newtheoremstyle{assumptionStyle}
  {0.5em}
  {0.5em}
  {}            
  {0.5cm}       
  {\bfseries}   
  {}            
  {0.2cm}       
  {}            

\newcommand{\thistheoremname}{}
\theoremstyle{assumptionStyle}
\newtheorem*{genericthm*}{\textup{\thistheoremname}}
\newenvironment{assumption}[2][1.6cm]{
	\renewcommand{\thistheoremname}{#2}%
	\begin{genericthm*}\hangindent=#1
	\setlength{\parindent}{#1+1.5em}} 
	{\end{genericthm*}}

\newcommand{\asNonInter}{\hyperlink{asNonInter}
        {\textcolor{black}{\textup{\textbf{(L)}}}}}
\newcommand{\asTB}{\hyperlink{asTB}
        {\textcolor{black}{\textup{\textbf{(TB)}}}}}
\newcommand{\asSC}{\hyperlink{asSC}
        {\textcolor{black}{\textup{\textbf{(SC)}}}}}
\newcommand{\asStab}{\hyperlink{asStab}
        {\textcolor{black}{\textup{\textbf{(STAB)}}}}}
\newcommand{\asREF}{\hyperlink{asREF}
        {\textcolor{black}{\textup{\textbf{(REF)}}}}}
\newcommand{\asGAP}{\hyperlink{asGAP}
        {\textcolor{black}{\textup{\textbf{(GAP)}}}}}
\newcommand{\asSTABref}{\hyperlink{asSTABref}
        {\textcolor{black}{\textup{\textbf{(STAB}}$^\mathrm{\bm{ref}}$\textbf{\textup{)}}}}}
\newcommand{\asP}{\hyperlink{asP}
        {\textcolor{black}{\textup{\textbf{(P)}}}}}
\newcommand{\asFF}{\hyperlink{asFF}
        {\textcolor{black}{\textup{\textbf{(FF)}}}}}

\newcommand{\MARGINWIDTH}{28.5mm}
\usepackage{geometry}
\usepackage{cleveref}
\crefformat{equation}{\textup{(#2#1#3)}}
\crefformat{section}{\S#2#1#3} 
\crefformat{subsection}{\S#2#1#3}
\crefformat{subsubsection}{\S#2#1#3}
\crefformat{thm}{Theorem~#2#1#3}

\newcommand{\refCE}[3]{$(\hyperref[#1]{\mathrm{CE}^{{#2},{#3}}})$}
\newcommand{\refGCE}[4]{$(\hyperref[#1]{\mathrm{GCE}_{{#4}}^{{#2},{#3}}})$}

\begin{document}

\title
    [Locality in Self-Consistent Tight Binding Models]
    {Locality of Interatomic Interactions in\\Self-Consistent Tight Binding Models}

\author{Jack Thomas}

\begin{abstract}
    A key starting assumption in many classical interatomic potential models for materials is a site energy decomposition of the potential energy surface into contributions that only depend on a small neighbourhood.
    Under a natural stability condition, we construct such a spatial decomposition for self-consistent tight binding models, extending recent results for linear tight binding models to the non-linear setting.
    %
    %
    %
    %
    %
    %
\end{abstract}

\maketitle

\section{Introduction}
Electronic structure models are widely used to calculate many optical, magnetic and mechanical properties of materials and molecules \cite{bk:finnis,bk:martin04}. Self-consistent (non-linear) tight binding models are simple electronic structure models that are interesting in their own right but also provide convenient prototypes for the much more complicated density functional theories. A paradigm example being the density functional tight binding (DFTB) method \cite{KoskinenVille2009dftb,SeifertJoswig2012dftb,ElstnerSeifert2014dftb}. 
\let\thefootnote\relax\footnote{
    \textit{Date:}~\today.\\
    \textit{2020 Mathematics Subject Classification:}~74E15, 81V45. \\
    \textit{Key words and phrases:}~strong locality; self-consistent tight binding; point defects; insulators.\\
    This work is supported by EPSRC as part of the MASDOC DTC, Grant No.~EP/HO23364/1.\\
    \href{mailto:j.thomas.1@warwick.ac.uk?subject=Non-linear Tight Binding}
    {\textcolor{black}{\texttt{J.Thomas.1@warwick.ac.uk.}}}~Mathematics Institute, Zeeman Building, University of Warwick, Coventry, UK.
}

In contrast to previous works on the linear tight binding model \cite{ChenLuOrtner18,ChenOrtner16,ChenOrtnerThomas2019:locality}, the self-consistency introduces the interesting issue of stability of the electronic structure problem. Therefore, under a suitable stability condition \cite{ELu10}, we show that the potential energy surface in this model can be decomposed into exponentially localised site contributions, thus justifying many classical interatomic potential (IP) and multi-scale methods. 

Despite the relative simplicity of tight binding models, a naive implementation demands $O(N^3)$ computational cost, where $N$ is the number of particles in the system.
It may therefore be advantageous to instead implement an IP model. In this case, the relevant parameters can be fitted to high accuracy by machine learning techniques together with theoretical data resulting from a high-fidelity model
\cite{Bartok2018, 
BartokCsanyiEtAl2010,
BehlerParrinello2007,
Shapeev2016}.
In most IP models for materials, a necessary starting assumption is that the potential energy surface can be decomposed into localised site contributions. That is, 
for atomic positions $y = \{y_n\}$, the total energy $E = E(y)$ may be written
\begin{equation}\label{eq:local_intro}
    E(y) = \sum_\ell E_\ell(y) 
    \qquad \text{with} \qquad 
    \left| \frac{\partial^j E_\ell(y)}{\partial y_{n_1}\dots\partial y_{n_j}}\right| 
        \lesssim e^{-\eta \sum_{l=1}^j |y_\ell - y_{n_l}|},
\end{equation}
for some $\eta > 0$. 

Typically, classical IP models are short-ranged and thus only justified if the exponent $\eta$ in \cref{eq:local_intro}, which measures the interatomic interaction range, is not too small. Also, in the context of QM/MM multi-scale methods, $\eta$ gives a guide for the size of the (computationally more expensive) QM region that must be imposed 
\cite{ChenOrtner16,
ChenOrtner17,
CsanyiAlbaretMorasPayneDeVita05}.
Therefore, in the present paper, not only are we interested in obtaining a site energy decomposition as in \cref{eq:local_intro}, but we also wish to describe the exponent.

In addition to the partial justification for IP and multi-scale models, our results also allow the thermodynamic and zero Fermi-temperature limit results of \cite{OrtnerThomas2020:pointdef,ChenLuOrtner18} to be extended to the non-linear setting. We sketch the main ideas in the concluding remarks of \S\ref{sec:conclusions}.

\subsection{Summary of Results}
The results of this paper are divided into two sections: we discuss general locality estimates in \S\ref{sec:results_general} and improve these results for the specific case of point defects in insulating multi-lattice materials in \S\ref{section:point_defects}. 

\subsubsection*{General Locality Estimates} 
The previous works \cite{ChenLuOrtner18,ChenOrtner16} provide partial justification for \cref{eq:local_intro} in the setting of linear tight binding models at finite Fermi-temperature, $T$. However, in general, we can only expect
\[
    \eta\sim T \quad \text{as} \quad T\to0
\]
in this case, meaning that for low temperature regimes the practical value of \cref{eq:local_intro} is limited. However, in the case of insulating multi-lattice materials (where there is a spectral gap in the system), the locality estimates are improved \cite{ChenOrtnerThomas2019:locality}, and extended to the zero Fermi-temperature case. In this setting, the exponent $\eta$ is linear in the spectral gap. 

In \S\ref{sec:results_general}, we simultaneously extend both \cite{ChenOrtner16,ChenLuOrtner18} (in the case of finite Fermi-temperature) and \cite{ChenOrtnerThomas2019:locality} (for insulators at finite or zero Fermi-temperature) to the non-linear setting.

\subsubsection*{Point Defects in Insulators}
Simulating local defects in materials remains an issue of great interest in the solid state physics and materials science communities \cite{Pisani1994,bk:Stoneham2001}. See \cite{CancesLeBris2013} for a mathematical review of some works related to the modelling of point defects in materials science. 

When considering point defects in the material, ``pollution'' of the spectrum may enter band gap which \textit{a priori} affects the exponent $\eta$ in \cref{eq:local_intro}. However, by approximating the defect as a low rank perturbation, we show in \S\ref{section:point_defects} that the locality results only weakly depend on the defect and the estimates resemble the defect-free estimates for sites away from the defect core. That is, the exponents in the locality estimates depend only on a local environment of the particular atomic site. This extends results of \cite{ChenOrtnerThomas2019:locality} to the non-linear setting.

\subsection{Notation}
For an operator $T$, the discrete spectrum (isolated eigenvalues of finite multiplicity) and essential spectrum are denoted by $\sigma_\mathrm{disc}(T)$ and $\sigma_\mathrm{ess}(T) \coloneqq \sigma(T) \setminus \sigma_\mathrm{disc}(T)$, respectively. 

For sequences $\psi$, the $j^\mathrm{th}$ entry is written $[\psi]_j$ and the $\ell^2$ norm is $\|\psi\|_{\ell^2}$. For bounded $\mathbb C$-valued functions, we denote by $\|\cdot\|_\infty$ the supremum norm. For matrices (or operators with matrix entries only non-zero on a finite sub-matrix) $M$, we denote by $\|M\|_\mathrm{max}\coloneqq \max_{ij} |M_{ij}|$ the maximum-norm. For operators $T$ on $\ell^2$, $\|\cdot\|_\mathrm{F}$ denotes the Frobenius or Hilbert-Schmidt norm and $\|\cdot\|_{\ell^2 \to \ell^2}$ the operator norm.

On $\mathbb R^n$ or $\mathbb C$, we will denote the Euclidean norm by $|\cdot|$ and the open balls of radius $\delta$ about $a$ and $0$ by $B_\delta(a)$ and $B_\delta \coloneqq B_\delta(0)$, respectively. For a subset $A$ of $\mathbb R^n$ or $\mathbb C$, we write 
$B_\delta(A) \coloneqq 
    \{ x \colon 
    \mathrm{dist}(x,A) < \delta \}$
where $\mathrm{dist}(b,A) \coloneqq \inf_{a\in A}|a - b|$. The Hausdorff distance between two sets $A$ and $B$ is denoted $\mathrm{dist}(A,B) \coloneqq \max\{ \sup_{a \in A} \mathrm{dist}(a, B), \sup_{b \in B} \mathrm{dist}(b, A) \}$.

We write $b + A = \{ b + a \colon a \in A \}$, $A - b \coloneqq \{a - b \colon a \in A\}$ and $r A \coloneqq \{r a \colon a \in A\}$. 
If $A$ is finite then $\#A$ denotes the cardinality of $A$. For an index set $A$, we denote by $\delta_{ij}$ the Kronecker delta function for $i,j \in A$. The set of non-negative real numbers will be denoted by $\mathbb R_+$.

The symbol $C$ will denote a generic positive constant that may change from one line to the next. In calculations, $C$ will always be independent of Fermi-temperature. The dependencies of $C$ will be clear from context or stated explicitly. When convenient to do so we write $f \lesssim g$ to mean $f \leq C g$ for some generic positive constant as above.

\section{Results: General Locality Estimates}
\label{sec:results_general}
\subsection{Tight Binding Model}            
\label{subsec:tight_binding_model}          

For a locally finite reference configuration $\Lambda \subset \mathbb R^d$ and displacement $u \colon \Lambda\to\mathbb R^d$, we write 
$\bm{r}_{\ell k}(u) \coloneqq \ell + u(\ell) - k - u(k)$
and 
$r_{\ell k}(u) \coloneqq |\bm{r}_{\ell k}(u)|$. 
We consider displacements $u$ satisfying the following uniform non-interpenetration condition:
\hypertarget{asNonInter}{\label{asNonInter}}
\begin{assumption}{\asNonInter}
There exists $\ctnoninterpen>0$ such that
$r_{\ell k}(u) \geq \ctnoninterpen|\ell - k|$ for all $\ell, k \in \Lambda$.
\end{assumption}

We consider $\numorbitals$ atomic orbitals per atom, indexed by $1\leq a,b \leq \numorbitals$. For given displacements $u$, we consider corresponding electronic densities, $\rho\colon \Lambda \to \mathbb R_+$, satisfying a self-consistency condition (see \asSC, below), which introduces the non-linearity to the model; the main departure from the previous works \cite{ChenOrtner16,ChenOrtnerThomas2019:locality}. 
We define the two-centre tight binding Hamiltonian as follows:\hypertarget{asTB}{\label{asTB}}
\begin{assumption}{\asTB}
    For $\ell, k \in \Lambda$ and $1\leq a,b \leq\numorbitals$, we suppose that the entries of the Hamiltonian take the form
    \begin{equation}\label{a:two-centre}
        \Ham(u;\rho)_{\ell k}^{ab} 
            \coloneqq h^{ab}_{\ell k}\left(\bm{r}_{\ell k}(u)\right) + v(\rho(\ell)) \delta_{\ell k} \delta_{ab}
    \end{equation}
    where $h^{ab}_{\ell k}\colon \mathbb R^d \to \mathbb R$ are $\ctHamregularity$ times continuously differentiable for some $\ctHamregularity\geq1$ and $v$ is a bounded smooth function on $(0,\infty)$ with bounded derivatives. 
    
    Further, we assume that there exist 
    $\ctTBprefactor, \ctTBexponent >0$
    such that, for each $1 \leq j \leq \ctHamregularity$, 
    \begin{equation}\label{a:TB}	
        \left|h^{ab}_{\ell k}(\xi)\right| \leq \ctTBprefactor\,e^{-\ctTBexponent |\xi|}
        \quad \text{and} \quad 
        \left|\partial^\alpha h_{\ell k}^{ab}(\xi)\right| \leq \ctTBprefactor \, e^{-\ctTBexponent |\xi|} 
        \quad \forall \xi \in \mathbb R^d
    \end{equation}
    for all multi-indices $\alpha \in \mathbb N^d$ with $|\alpha|_1 = j$. 
    
    Finally, we suppose that $h^{ab}_{\ell k}(\xi) = h_{k\ell}^{ba}(-\xi)$ for all $\xi \in \mathbb R^d$, $1\leq a,b\leq\numorbitals$ and $\ell, k \in \Lambda$.
\end{assumption}
The constants $h_0$ and $\gamma_0$ in \cref{a:TB} are independent of the atomic sites. 

The symmetry assumption in \asTB~means that the Hamiltonian, $\Ham(u;\rho)$, is symmetric and thus the spectrum is real. Moreover, $\sigma(\Ham(u;\rho)) \subset [\underline{\sigma},\overline{\sigma}]$ for some $\underline{\sigma},\overline{\sigma}$ depending on $\mathfrak{m}, d, h_0, \gamma_0, \|v\|_\infty$ and are independent of the size of the system and the displacement $u$ satisfying \asNonInter~with the constant $\mathfrak{m}$. In fact, by generalising the proof of \cite[Lemma~4]{ChenOrtner16} to the setting we consider here, we obtain $\sigma(\Ham(u;\rho)) \subset \|v\|_\infty + C_d h_0 (\gamma_0 \mathfrak{m})^{-d} [-1,1]$.

The pointwise bound on $|h^{ab}_{\ell k}|$ in \cref{a:TB} is more general than many tight binding models which impose a finite cut-off radius. The assumption for $|\alpha| = 1$, states that there are no long range interactions in the model and so, in particular, we assume that Coulomb interactions have been screened, which is typical in practical tight binding models \cite{cohen94,mehl96,papaconstantopoulos97}.

Following 
\cite{ChenOrtner16,
      ChenLuOrtner18,
      OrtnerThomas2020:pointdef,
      ChenOrtnerThomas2019:locality},
we consider finite energy displacements. For $u \colon \Lambda \to \mathbb R^d$, $\ell \in \Lambda$ and $\sigma \in \Lambda - \ell$, we define the finite difference $D_\sigma u(\ell) \coloneqq u(\ell + \sigma) - u(\ell)$ and the full finite difference stencil $Du(\ell) \coloneqq (D_\sigma u(\ell))_{\sigma \in \Lambda - \ell}$. We then let $\dot{\mathscr{W}}^{1,2}(\Lambda)$ be the set of finite energy displacements: for $\ctGamma > 0$,
\[
    \dot{\mathscr{W}}^{1,2}(\Lambda) 
    \coloneqq \bigg\{ u \colon \Lambda \to \mathbb R^d \colon
        \|Du\|_{\ell^2_\ctGamma}^2 \coloneqq 
            \sum_{\ell \in \Lambda} 
            \sum_{\sigma \in \Lambda - \ell} 
            e^{-2\ctGamma |\sigma|} |D_\sigma u(\ell)|^2 < \infty 
    \bigg\}
\]
which is well defined since all the semi-norms $\|D\cdot\|_{\ell^2_\ctGamma}$ are all equivalent for $\ctGamma>0$ \cite{ChenNazarOrtner19}. For the remainder of this paper, we fix an exponent $\ctGamma>0$. In this section, we require perturbations to have finite energy (see Lemma~\ref{lem:perturbation}, below) whereas, when considering the improved estimates in the case of point defects in \S\ref{section:point_defects}, we also require the configurations to be given by finite energy displacements. 

\subsection{Spatial Decomposition of Quantities of Interest} 

Let $u \colon \Lambda \to \mathbb R^d$ satisfy \asNonInter~and suppose $\rho$ is an associated electronic density.

\subsubsection{Local Quantities of Interest}

We suppose $\mathfrak{o}\colon \mathbb R\to \mathbb R$ is a function that extends analytically to an open neighbourhood of $\sigma(\Ham(u;\rho))$ in $\mathbb C$ and define the corresponding local quantities, $O_\ell$ (for $\ell \in \Lambda$), by 
\begin{align}
\label{eq:local_integral}
    O_\ell(u;\rho) \coloneqq 
        - \frac{1}{2\pi i} \sum_a 
        \oint_{\mathscr C_\mathfrak{o}} \mathfrak{o}(z) 
            \left[ 
                ( \Ham(u;\rho) - z )^{-1}
            \right]_{\ell\ell}^{aa}
        \mathrm{d}z
\end{align}
where $\mathscr C_\mathfrak{o}$ is a simple, closed, positively oriented contour contained within the region of holomorphicity of $\mathfrak{o}$, encircling $\sigma(\Ham(u;\rho))$ and such that
\begin{align}
\label{eq:do}
    \mathsf{d}_\mathfrak{o} \coloneqq 
    \min_{z \in \mathscr C_\mathfrak{o}} 
    \mathrm{dist}\big(
            z, \sigma(\Ham(u;\rho))
        \big) > 0.
\end{align}

For finite systems, we may diagonalise the Hamiltonian, $\Ham(u;\rho) = \sum_s \lambda_s \ket{\psi_s}\bra{\psi_s}$ (where $(\lambda_s,\psi_s)$ are normalised eigenpairs), and note that many quantities of interest, including the Helmholtz free energy, grand potential and the particle number functional \cite{ChenLuOrtner18,ChenOrtner16}, may be written as a sum of the local contributions \cref{eq:local_integral} for some appropriate choice of $\mathfrak{o}$:
\begin{align}
\label{eq:total_energy}
    \sum_{\ell \in \Lambda} O_\ell(u;\rho)
    = \sum_{\above
                {\ell \in \Lambda}
                {1\leq a\leq\numorbitals}
            } 
      \sum_s \mathfrak{o}(\lambda_s) [\psi_s]_{\ell a}^2
    = \sum_s \mathfrak{o}(\lambda_s).
\end{align}
This decomposition is well known and follows from a spatial decomposition of the total density of states \cite{Ercolessi2005,ChenOrtner16,bk:finnis,ChenLuOrtner18}.

Our motivation comes from viewing \cref{eq:total_energy} as the total energy of the system and \cref{eq:local_integral} as a spatial decomposition of this energy. The aim of the present paper is to show the exponential localisation of \cref{eq:local_integral} with respect to perturbing atomic positions, which, in the case of site energies, justifies IP and multi-scale models as discussed in the introduction.

\subsubsection{Self-consistency}
After fixing an inverse Fermi-temperature $\beta >0$, we define 
$f(z - \mu) = (1 + e^{\beta(z - \mu)})^{-1}$
to be the Fermi-Dirac occupation distribution where $\mu$ is a fixed chemical potential. We let $\mathscr C_f$ be a simple closed contour encircling $\sigma(\Ham(u;\rho))$ and avoiding the singularities of $f(\,\cdot - \mu)$. That is, avoiding 
$\mu + i \pi\beta^{-1} \left( 2\mathbb Z + 1\right)$
and such that 
$\mathsf{d}_f \coloneqq 
    \min_{z \in \mathscr C_f} \mathrm{dist}
    \left(
        z, \sigma(\Ham(u;\rho))
    \right) > 0$.
%
In general, we may choose $\mathscr C_f$ so that 
$\mathsf{d}_{f} \geq \tfrac{\pi}{2\beta}$.
However, for insulators, as we shall see in \S\ref{section:point_defects}, this constant may be chosen to be linear in the spectral gap at $\mu$.

For the case of zero Fermi-temperature, we take the pointwise $\beta \to \infty$ limit and define 
$f(z-\mu) \coloneqq \chi_{(-\infty,\mu)}(z) + \frac{1}{2}\chi_{\{\mu\}}(z)$.
For insulating systems, there is a spectral gap at $\mu$ 
(that is, with $\mathsf{g} \coloneqq \inf [\mu,+\infty) \cap \sigma(\Ham(u;\rho)) - \sup (-\infty, \mu] \cap \sigma(\Ham(u;\rho)) > 0$)
and so $f(\,\cdot\, - \mu)$ is analytic in a neighbourhood of $\sigma(\Ham(u;\rho))$. In this case, we let $\mathscr C_f$ be a simple closed contour encircling 
$\sigma(\Ham(u;\rho)) \cap (-\infty, \mu)$
and avoiding 
$\sigma(\Ham(u;\rho)) \cap [\mu, \infty)$
with 
$\mathsf{d}_f = 
    \min_{z \in \mathscr C_f}\mathrm{dist}
    \left(
        z, \sigma(\Ham(u;\rho))
    \right) \geq \frac{1}{2}\mathsf{g}$.

For fixed $\beta \in (0,\infty]$ as above, we can therefore define %
\begin{align}\label{eq:Fl}
    F_\ell(u;\rho) \coloneqq O_\ell(u;\rho) 
    \quad \textup{with} \quad 
    \mathfrak{o} = f(\,\cdot - \mu) 
    \quad \textup{and} \quad
    \mathscr C_\mathfrak{o} = \mathscr C_f
\end{align}
as in \cref{eq:local_integral}. To simplify notation in the following, we will write $F(u;\rho) = (F_\ell(u;\rho))_{\ell \in \Lambda}$.

For a given displacement, we consider corresponding self-consistent electronic densities, giving rise to the non-linearity of the problem:
\hypertarget{asSC}{\label{asSC}}
\begin{assumption}[1.55cm]{\asSC}
    For a displacement $u\colon \Lambda \to \mathbb R^d$ satisfying \asNonInter, we say that $\rho$ is an associated self-consistent electronic density if $\rho = F(u;\rho)$.
\end{assumption}

\begin{remark}[Self-consistency]
    For a finite system, the self-consistency equation \asSC~takes the following form:
    \begin{align}
        \rho(\ell) = \sum_{s,a} f(\lambda_s - \mu) [\psi_s]_{\ell a}^2
    \end{align}
    where 
    $\Ham(u;\rho) =\sum_s \lambda_s \ket{\psi_s}\bra{\psi_s}$
    for normalised eigenpairs $(\lambda_s, \psi_s)$. That is, the electronic structure of the system is obtained by assigning electrons to the eigenstates of lowest energy, according to the Fermi-Dirac occupation distribution and subject to Pauli's exclusion principle. 
\end{remark}

\subsubsection{Stability}

We wish to show that, for fixed $u\colon \Lambda \to \mathbb R^d$ and an associated self-consistent electronic density $\rho$, the quantities $O_\ell(u;\rho)$ are exponentially localised. As shown in \cite{ChenOrtnerThomas2019:locality} for the linear model (that is, neglecting the $v(\rho(\ell)) \delta_{\ell k}$ term in \cref{a:two-centre}), the exponent in these locality results are linear in $\mathsf{d}_\mathfrak{o}$ from \cref{eq:do}. For the more complicated, non-linear model that we consider here, the locality also depends on the stability of the model; discussed below.

Supposing that $(u,\rho)$ satisfies a natural stability condition (see \asStab, below), it is possible to rewrite the local quantities of interest as a function of the displacements. That is, for $\widetilde{u}$ in a neighbourhood of $u$, there exists a locally unique $\widetilde{\rho} = \widetilde{\rho}(\widetilde{u})$ in a neighbourhood of $\rho$ such that $(\widetilde{u},\widetilde{\rho}(\widetilde{u}))$ satisfies \asSC~and we can therefore write
\begin{align}
\label{eq:local_integral_sc}
    O_\ell^\mathrm{sc}(\widetilde{u}) \coloneqq 
        O_\ell\big(\widetilde{u};\widetilde{\rho}(\widetilde{u})\big).
\end{align}
Moreover, the mapping $\widetilde{u} \mapsto O^\mathrm{sc}_\ell(\widetilde{u})$ is $\nu$ times continuously differentiable in a neighbourhood of $u$. See Lemma~\ref{lem:perturbation} for the rigorous statement. 

We can now consider the derivatives of the local quantities of interest with respect to the perturbation of atomic positions. Using \cref{eq:local_integral}, it is sufficient to consider the derivatives of the resolvent operators. 
Since the linear contribution has been studied in \cite{ChenOrtnerThomas2019:locality}, we are only concerned with the additional non-linear part, which involves derivatives of the electronic density. Due to the self-consistency, we obtain the formula
\begin{align}\label{eq:stab_motivation}
    \frac
        { \partial \rho(\ell) }
        { \partial [u(m)]_i } 
    =   \left[
            (I - \mathscr L(u;\rho))^{-1} \phi^{(m)} 
        \right]_\ell
\end{align}
where the \textit{stability operator}, $\mathscr L(u;\rho)$, is the Jacobian of $F(u;\rho)$ with respect to $\rho$ and $\phi^{(m)} \in \ell^2(\Lambda)$. 
%
%
Therefore, 
the following stability condition, which we take from \cite{ELu10,ELu2012}, is the minimal starting assumption required for the analysis:
\hypertarget{asStab}{\label{asStab}}
\begin{assumption}[2.1cm]{\asStab}
    We say $(u,\rho)$ is stable if $I - \mathscr L(u;\rho)$ is invertible as an operator on $\ell^2(\Lambda)$ where $\mathscr L(u;\rho)$ is the Jacobian of $F(u;\rho)$ with respect to $\rho$. For a stable configuration $(u,\rho)$, we write
    \begin{align}\label{eq:dstab}
        \mathsf{d}_\mathscr{L} \coloneqq \mathrm{dist}(1, \mathscr{L}(u;\rho) ) > 0.
    \end{align}
\end{assumption}
\begin{remark}
    \textit{(i)} Equivalently, \cref{eq:dstab} states that $\|\left( I - \mathscr L(u;\rho) \right)^{-1} \|_{\ell^2\to\ell^2} \leq \mathsf{d}_\mathscr{L}^{-1}$. 
    
    \textit{(ii)} A simple calculation reveals that 
    $\mathscr L(u;\rho)\colon \ell^2(\Lambda) \to \ell^2(\Lambda)$
    has matrix entries
    \begin{align}\label{eq:L}
        \mathscr{L}(u;\rho)_{\ell k} 
            = \frac{1}{2\pi i} \oint_{\mathscr C_f} f(z-\mu) 
                \sum_{a,b = 1}^{\numorbitals}
                \left(
                    \left[
                        \left( \Ham(u;\rho) - z \right)^{-1}
                    \right]^{ab}_{\ell k}
                \right)^2 
            \mathrm{d}z \, v^\prime(\rho(k)).
    \end{align}
\end{remark}

\subsubsection{Locality}

We are now in a position to state general locality estimates for $\beta \in (0,\infty)$ and for insulators in the case $\beta = \infty$:
\begin{theorem}[General Locality Estimates]
\label{thm:locality}
    Suppose $(u,\rho)$ satisfies \asNonInter, \asSC, \asStab~and let $\mathsf{d}_{\mathfrak{o}}$, $\mathsf{d}_{f}$, $\mathsf{d}_\mathscr{L}>0$ be the constants from \cref{eq:do} and \cref{eq:dstab}.
    Then, for $1 \leq j \leq \ctHamregularity$, $\ell \in \Lambda$, $\bm{m} = (m_1,\dots,m_j) \in \Lambda^j$ and $1\leq i_1,\dots,i_j \leq d$, there exists $C_{j}>0$ such that
    \begin{align*}
        \left| 
            \frac
                {\partial^j O^\mathrm{sc}_{\ell}(u)}
                {\partial [u(m_1)]_{i_1}\dots \partial [u(m_j)]_{i_j}}
        \right|
        \leq C_{j} e^{-\eta \sum_{l=1}^j r_{\ell m_l}(u) }
    \end{align*}
    where $\eta \coloneqq c \min\left\{ 1, \mathsf{d}_{\mathfrak{o}}, \mathsf{d}_{f}, c_f \mathsf{d}_{\mathscr L}\right\}$, $c_f \coloneqq \mathsf{d}_{f}^2 \min\{ 1, \mathsf{d}_{f}^{d+1} \}$ and $c$ is a positive constant depending on 
    $\gamma_0$, $h_0$, $\mathfrak{m}$, $\numorbitals$, $d$, $j$, $\|Du\|_{\ell^2_\ctGamma}$, $\|v^\prime\|_\infty$ and on the length of $\mathscr C_f$.
\end{theorem}
\begin{proof}
    The proof of this result follows the analogous proof in the linear case \cite{ChenOrtnerThomas2019:locality}, together with bounds on the non-linear contribution \cref{eq:stab_motivation}. Full details are presented in \S\ref{sec:proofs}.
\end{proof}

\section{Results: Point Defects in Insulating Materials}
\label{section:point_defects}
Now we consider the specific example of point defect reference configurations. In this case we show ``improved'' locality estimates in which the pre-factors and exponents behave like the corresponding reference constants.  

            \subsection{Point Defect Configurations}                

We suppose that $\Lambda^\mathrm{ref} \subset \mathbb R^d$ is multi-lattice and $\rho^\mathrm{ref}$ is a corresponding self-consistent electronic density:
\hypertarget{asREF}{\label{asREF}}
\begin{assumption}[1.85cm]{\asREF}
    We suppose that there exists a non-singular matrix $\mathsf{A} \in \mathbb R^{d\times d}$ and a {unit cell} $\Gamma \subset \mathbb R^d$ such that $\Gamma$ is finite, contains the origin and
    \[
        \Lambda^\mathrm{ref} 
            \coloneqq \bigcup_{\gamma \in \mathbb Z^d} 
            \left(
                \Gamma + \mathsf{A}\gamma
            \right).
    \]
    Moreover, we require the Hamiltonian to satisfy the following translational invariance property: 
    $h^{ab}_{\ell + \mathsf{A}\gamma_1, k + \mathsf{A}\gamma_2}(\xi) = h^{ab}_{\ell k}(\xi)$
    for all $\gamma_1, \gamma_2 \in \mathbb Z^d$, $\xi \in \mathbb R^d$ and $1\leq a,b \leq \numorbitals$. Further, we suppose $\rho^\mathrm{ref}$ is a translational invariant, self-consistent electronic density: 
    $(x^\mathrm{ref}, \rho^\mathrm{ref})$
    satisfies \asSC~where 
    $x^\mathrm{ref}\colon\Lambda^\mathrm{ref}\to\Lambda^\mathrm{ref}$
    is the identity configuration on 
    $\Lambda^\mathrm{ref}$ and 
    $\rho^\mathrm{ref}(\ell + \mathsf{A}\gamma) = \rho^\mathrm{ref}(\ell)$ 
    for all $\ell \in \Gamma$ and $\gamma \in \mathbb Z^d$.
\end{assumption}

It is possible that there are many electronic densities $\rho^\mathrm{ref}$ satisfying the conditions of \asREF. However, this is an issue that we will not concern ourselves with here, and, for the remainder of this paper, we will simply fix any electronic density $\rho^\mathrm{ref}$ satisfying \asREF.

The translational invariance property of the Hamiltonian states that all $\ell + \mathsf{A}\gamma$ are of the same atomic species.

To simplify notation, when we consider the reference configuration ($\Lambda = \Lambda^\mathrm{ref}$), we will write,
\begin{gather*}
    \Ham^\mathrm{ref} \coloneqq \Ham(\bm{0};\rho^\mathrm{ref})
    \quad \textup{and} \quad 
    \mathscr{L}^\mathrm{ref} \coloneqq \mathscr{L}(\bm{0};\rho^\mathrm{ref}).
\end{gather*}
By exploiting the translational invariance of the the reference configuration, and by use of the Bloch transform \cite{bk:kittel}, we may conclude that $\sigma(\Ham^\mathrm{ref})$ and $\sigma(\mathscr L^\mathrm{ref})$ can be written as a union of finitely many spectral bands:
\begin{align}
    \sigma(\Ham^\mathrm{ref}) = \bigcup_n \lambda_n(\Gamma^\star) 
    \qquad \textup{and} \qquad
    \sigma(\mathscr L^\mathrm{ref}) = \bigcup_n \varepsilon_n(\Gamma^\star) 
\end{align}
where $\Gamma^\star \subset \mathbb R^d$ is a compact, connected set and $\lambda_n$ and $\varepsilon_n$ are continuous functions. Full details follow the calculations of \cite{ELu10} and are given in Appendix~\ref{app:bands} for completeness. 

In the following, we consider insulating materials and thus assume that there is a band gap in the reference Hamiltonian:
\hypertarget{asGAP}{\label{asGAP}}
\begin{assumption}{\asGAP}
    We assume $\mu \not\in \sigma(\Ham^\mathrm{ref})$ and define 
    \[
        \mathsf{g}^\mathrm{ref} \coloneqq 
            \inf
                \left( 
                    \sigma(\Ham^\mathrm{ref}) \cap (\mu,\infty) 
                \right) - 
            \sup
                \left( 
                    \sigma(\Ham^\mathrm{ref}) \cap (-\infty,\mu) 
                \right)
            > 0.
    \]
\end{assumption}
Therefore, in the finite temperature case ($\beta < \infty$), we may consider a contour $\mathscr C_f$ as in \cref{eq:Fl} with $\mathsf{d}^\mathrm{ref}_f \coloneqq \min_{z \in \mathscr C_f} \mathrm{dist}(z,\sigma(\Ham^\mathrm{ref})) \geq \frac{1}{2}\mathrm{dist}\left(\mu, \sigma(\Ham^\mathrm{ref})\right)$. In particular, we have $\mathsf{d}^\mathrm{ref}_f \not\to 0$ in the zero temperature limit. On the other hand, for zero temperature ($\beta = \infty$), we may choose $\mathscr C_f$ as in \cref{eq:Fl} with $\mathsf{d}^\mathrm{ref}_f = \frac{1}{2}\mathsf{g}^\mathrm{ref}$.

Further, we suppose that the reference configuration is stable:
\hypertarget{asSTABref}{\label{asSTABref}}
\begin{assumption}{\asSTABref}
    $\exists \,\mathsf{d}_{\mathscr L}^{\mathrm{ref}}>0$ 
    such that 
    $I - \mathscr L^\mathrm{ref}$
    is invertible with 
    $\|(I - \mathscr L^\mathrm{ref})^{-1}\|_{ \ell^2 \to \ell^2 } 
        \leq (\mathsf{d}_{\mathscr L}^{\mathrm{ref}})^{-1}$.
\end{assumption}
Now, given a reference, $\Lambda^\mathrm{ref}$, as above, we consider point defect configurations, $\Lambda$, satisfying:
\hypertarget{asP}{\label{asP}}
\begin{assumption}{\asP}
    There exists $R_\mathrm{def}>0$ such that 
    $\Lambda \setminus B_{R_\mathrm{def}} = \Lambda^\mathrm{ref} \setminus B_{R_\mathrm{def}}$
    and $\Lambda \cap B_{R_\mathrm{def}}$ is finite.
\end{assumption}

We will consider electronic densities $\rho: \Lambda \to \mathbb R_+$ satisfying the following mild technical assumption on the far-field behaviour:
\hypertarget{asFF}{\label{asFF}}
\begin{assumption}{\asFF}
    $
        \lim\limits_{|\ell| \to \infty}|v(\rho(\ell)) - v(\rho^\mathrm{ref}(\ell))| = 0
    $
    where $\rho^\mathrm{ref}$ is the fixed electronic density satisfying \asREF.
\end{assumption}
This assumption is explained in more detail in Remark~\ref{rem:FF}, below.

We now restrict the class of admissible configurations by considering finite energy displacements, $u \in \mathscr{W}^{1,2}(\Lambda)$, and show that, for such displacements, the spectra can be described in terms of the reference spectra.
\begin{lemma}[Perturbation of the Spectrum]
\label{eq:perturbation_spectrum}
    Let $u \in \dot{\mathscr W}^{1,2}(\Lambda)$ satisfy $\asNonInter$ and $\rho$ be an associated electronic density satisfying \asSC~and \asFF. 
    %
    %
    Then, for all $\delta > 0$, there exists $S_\delta$ such that 
    $\#\left(
            \sigma( \Ham(u;\rho) ) \setminus 
                B_\delta( \sigma(\Ham^\mathrm{ref}) )
        \right) 
        + 
    \#\left(
        \sigma( \mathscr L(u;\rho) ) \setminus 
            B_\delta( \sigma(\mathscr{L}^\mathrm{ref}) )
    \right) \leq S_\delta$.
\end{lemma}
\begin{proof}
We may apply Lemmas~\ref{lem:Ham_perturbation} and \ref{lem:stability_perturbation} together with \cite[Proof of Lemma~3]{ChenOrtnerThomas2019:locality} to conclude.
\end{proof}
Supposing that \asGAP~is satisfied, Lemma~\ref{eq:perturbation_spectrum} states that there are at most finitely many isolated eigenvalues lying inside the band gap and bounded away from the band edges.

        \subsection{Improved Locality Estimates}        

Just as in the linear case \cite{ChenOrtnerThomas2019:locality}, a Combes-Thomas resolvent estimate \cite{CombesThomas1973} applied to the spectral projections corresponding to the finitely many eigenvalues bounded away from the edges of the bands, together with a finite rank update formula (\textit{i.e.}~the Woodbury identity), allows us to approximate
    $(\Ham(u;\rho) - z)^{-1}$ and $(I - \mathscr L(u;\rho))^{-1}$ 
by finite rank updates of the reference resolvents 
    $(\Ham^\mathrm{ref} - z)^{-1}$ and $(I - \mathscr L^\mathrm{ref})^{-1}$,
respectively. Therefore, by applying the locality estimates of Theorem~\ref{thm:locality} on the reference spectrum, we obtain the following improved estimates for point defect configurations:

\begin{theorem}[Improved Locality Estimates]
\label{thm:improved_locality}
    Suppose that the reference configuration $\Lambda^\mathrm{ref}, \rho^\mathrm{ref}$ satisfies \asREF, \asGAP~and \asSTABref. Moreover, we fix $\Lambda$ satisfying \asP,
    $u\in\dot{\mathscr W}^{1,2}(\Lambda)$ and $\rho$ 
    satisfying \asNonInter, \asSC, \asStab~and \asFF. %
    %
    Then,
    \begin{itemize}
        \item[(i)] for $1 \leq j \leq \ctHamregularity$, $\ell \in \Lambda$, $\bm{m} = (m_1,\dots,m_j) \in \Lambda^j$ and $1\leq i_1,\dots,i_j \leq d$, there exists positive constants $C_j(\ell,\bm{m}) = C_j$, $\eta_j = \eta_j(\ell,\bm{m})$ such that
        \begin{align*}
            \left|\frac
                    {\partial^j O^\mathrm{sc}_{\ell}(u)}
                    {\partial [u(m_1)]_{i_1} \dots \partial [u(m_j)]_{i_j}} 
            \right|
            \leq C_{j} e^{-\eta_j \sum_{l=1}^j r_{\ell m_l}(u) }
        \end{align*}
        where 
        $\eta_j \coloneqq c_j \min
            \left\{ 
                1, 
                \mathsf{d}^\mathrm{ref}_{\mathfrak{o}}, 
                \mathsf{d}^\mathrm{ref}_{f}, 
                c_{f}(\ell,\bm{m}) \mathsf{d}^\mathrm{ref}_{\mathscr L}
            \right\}$,
        $c_j>0$ depends on 
            $j$,
            $\gamma_0$, 
            $h_0$, 
            $\mathfrak{m}$, 
            $d$, 
            $\|v^\prime\|_\infty$ 
        and on the lengths of 
            $\mathscr C_\mathfrak{o}$,
            $\mathscr{C}_f$ and 
        $c_{f}$ is a constant depending only on 
            $\ell$, 
            $\bm{m}$, 
            ${\mathsf{d}^\mathrm{ref}_{f}}$ and 
            ${\mathsf{d}_{f}}$.
        
        \item[(ii)] $C_j(\ell,\bm{m})$ is uniformly bounded and $c_{f}(\ell,\bm{m})$ is uniformly bounded away from zero independently of $\ell$ and $\bm{m}$. Let 
            $C^\mathrm{ref}_j\coloneqq C_j(\ell,\bm{m})$ and 
            $c^\mathrm{ref}_{f}\coloneqq c_{f}(\ell,\bm{m})$ 
        when 
        $\Lambda = \Lambda^\mathrm{ref}$, 
        $u = 0$ and 
        $\rho = \rho^\mathrm{ref}$. 
        If $\ell,m_1, \dots, m_j \in B_R(\xi)$ for some $R > 0$, then 
        $C_j(\ell,\bm{m}) \to C_j^\mathrm{ref}$ 
        and 
        $c_{f}(\ell,\bm{m}) \to c_{f}^\mathrm{ref}$ 
        as $|\xi| \to \infty$, with exponential rates.
    \end{itemize}
\end{theorem}
\begin{remark}\label{rem:FF}
    Here, we briefly give examples of when the assumption \asFF~is
    %
    %
    satisfied and we show that it is equivalent to 
    $\|\rho - \rho^\mathrm{ref}\|_{\ell^2} < \infty$.
    
    \textit{(i) Derivative of $v$ sufficiently small.} If $\|v^\prime\|_\infty$ is sufficiently small, we may treat the non-linear tight binding model as a non-linear perturbation of the corresponding linear model and would thus expect the main assumption
    $|v(\rho(\ell)) -v(\rho^\mathrm{ref}(\ell))| \to 0$ 
    to be satisfied. In fact, we now show that if 
    $\|v^\prime\|_\infty$
    is sufficiently small, then the stronger condition 
    $\|\rho - \rho^\mathrm{ref}\|_{\ell^2} < \infty$
    is satisfied. Using the Combes-Thomas estimates (Lemma~\ref{lem:CT}) for the resolvents we can conclude that there exists $\eta > 0$ such that
    \begin{align}\label{remark:rho_rho_ref}
    \begin{split}
        &\left| \rho(\ell) - \rho^\mathrm{ref}(\ell) \right|
        \leq C \Big| 
                    \sum_a \oint_{\mathscr{C}_f} f(z-\mu) 
                    \big[ 
                        \mathscr R_z(u;\rho) - \mathscr R_z^\mathrm{ref} 
                    \big]^{aa}_{\ell \ell} \mathrm{d}z 
                \Big| \\
        &\quad\leq C \sup_{z \in \mathscr C_f} \sum_a 
            \Big| 
                \big[ \mathscr R_z(u;\rho) 
                    \big( 
                        \Ham^\mathrm{ref} - \Ham(u;\rho) 
                    \big) \mathscr R_z^\mathrm{ref} 
                \big]^{aa}_{\ell \ell}  
            \Big| \\
        &\quad\leq  C \sum_{\ell_1,\ell_2} 
            e^{-\eta \left(
                        |\ell - \ell_1| + |\ell_1 - \ell_2| + |\ell_2 - \ell|
                    \right) } 
            |D_{\ell_2 - \ell_1} u(\ell_1)| + 
                    C \sum_{\ell_1} 
            e^{-2\eta |\ell - \ell_1|} 
                |v(\rho(\ell_1)) - v(\rho^\mathrm{ref}(\ell_1))|
    \end{split}
    \end{align}
    where $\mathscr R_z(u;\rho) \coloneqq (\Ham(u;\rho) - z)^{-1}$ and 
    $\mathscr R^\mathrm{ref}_z \coloneqq (\Ham^\mathrm{ref} - z)^{-1}$.
    Here, we have abused notation as the operators $\mathscr R_z(u;\rho)$ and $\mathscr R_z^\mathrm{ref}$ are defined on different spatial domains. This issue is resolved in \cite[\S4.3]{ChenOrtnerThomas2019:locality} and also briefly explained on page~\pageref{pg:spatial_domains}. After squaring \cref{remark:rho_rho_ref} and summing over $\ell \in \Lambda$, we obtain
    \begin{align*}
        \|\rho - \rho^\mathrm{ref}\|_{\ell^2}^2 
        &\leq 
        C_1+ {C}_2\|v^\prime\|^2_\infty \| \rho - \rho^\mathrm{ref} \|_{\ell^2}^2.
        %
        %
    \end{align*}
    Therefore, if $C_2\|v^\prime\|^2_\infty < 1$, then $\|\rho - \rho^\mathrm{ref}\|^2_{\ell^2} \leq C_1(1 - C_2\|v^\prime\|_\infty^2)^{-1}$. 
    
    \textit{(ii) Stability of the electronic structure.} Another approach involves integrating along a path between $\rho$ and $\rho^\mathrm{ref}$. In order to compare $u$ and reference configuration, we must assume that \asFF~is satisfied for
    %
    %
    $\rho(\bm{0})$, a self-consistent electronic density associated with the identity configuration on $\Lambda$. By the translational invariance of the Hamiltonian 
    (i.e.~for all $c\in\mathbb R^d$, $\Ham(u;\rho) = \Ham(u + c;\rho)$ where $(u+c)(\ell) = u(\ell) + c$),
    we obtain translational invariance of the quantities of interest (as in \cite{ChenLuOrtner18}). In particular, the quantities of interest may be written as functions of the finite difference stencil $Du(\ell)$ for some $\ell \in \Lambda$. Therefore, the electronic density solving $\rho = F(u;\rho)$ can also be written as a function of $Du(\ell)$. Now since 
    $\left|\frac
            {\partial \rho(\ell)}
            {\partial u(m)}
    \right| \lesssim e^{-\eta r_{\ell m}}$
    (see \cref{eq:d_rho_u_bound}), we formally obtain
    \begin{align}\label{eq:rho_minus_rho_ref}
    \begin{split}
        |v(\rho(\ell)) - v(\rho^\mathrm{ref}(\ell))| 
        &\leq  \|v^\prime\|_\infty 
                \left| 
                    \int_0^1 \sum_{\sigma \in \Lambda - \ell} 
                    \frac
                        {\partial \rho_t(\ell)}
                        {\partial D_\sigma u(\ell)}
                    \cdot D_\sigma u(\ell) \mathrm{d}t
                \right| 
                + |v(\rho(\bm{0})(\ell)) - v(\rho^\mathrm{ref}(\ell))| \\
        &\lesssim \sum_{\sigma \in \Lambda - \ell} 
            e^{-\eta|\sigma|} |D_\sigma u(\ell)| 
            + |v(\rho(\bm{0})(\ell)) - v(\rho^\mathrm{ref}(\ell))|
    \end{split}
    \end{align}
    where $\rho_t \coloneqq \rho(tDu(\ell))$. Therefore, by taking $\ell \to \infty$ we may conclude.
    However, in \cref{eq:rho_minus_rho_ref}, we have assumed that along the linear path between $u$ and $\bm{0}$, the electronic density is a well-defined differentiable function of the displacement. Generalising the argument above, we only need to assume that there exists a sequence of displacements such that we can integrate along a piecewise linear path between $u$ and $\bm{0}$. That is, along the path, we need unique self-consistent electronic densities, the uniform non-interpenetration assumption to be satisfied with a uniform constant and, in the case of zero Fermi-temperature, the spectrum of the Hamiltonian must avoid the chemical potential. 
    
    \textit{(iii) An equivalent assumption.} We claim that \asFF~is 
    %
    equivalent to the (\textit{a priori} stronger) condition that 
        $\|\rho - \rho^\mathrm{ref}\|_{\ell^2} < \infty$. 
    Indeed, by assuming that 
        $|v(\rho(\ell)) - v(\rho^\mathrm{ref}(\ell))| \to 0$, 
    the diagonal operator defined by 
    $D_{\ell \ell} \coloneqq v(\rho(\ell)) - v(\rho^\mathrm{ref}(\ell))$ 
    is compact and so Lemma~\ref{lem:Ham_perturbation} (given below) allows us to approximate the Hamiltonian $\Ham(u;\rho)$ with a finite rank update of $\Ham^\mathrm{ref}$. We can therefore use \cref{remark:rho_rho_ref} to obtain the following stronger bound: for all $\delta > 0$, there exists a Hilbert-Schmidt operator $P^\delta$ such that 
    $\|P^\delta\|_\mathrm{F} \leq \delta$ 
    and 
    $|\rho(\ell) - \rho^\mathrm{ref}(\ell)| 
        \leq C_\delta e^{-\eta |\ell|} + P^\delta_{\ell\ell}$.
    This is an argument similar to \cite[(4.18)$-$(4.20)]{ChenOrtnerThomas2019:locality}.
    
    We have simply written 
        $|v(\rho(\ell)) -v(\rho^\mathrm{ref}(\ell))| \to 0$ 
    as an assumption in Lemma~\ref{eq:perturbation_spectrum} and Theorem~\ref{thm:improved_locality} to simplify the presentation and avoid the technical issues detailed above. We briefly remark here that this is the minimal assumption needed for our analysis to hold. Indeed, if \asFF~is not satisfied, then the operator $\Ham(u;\rho) - \Ham^\mathrm{ref}$ is not compact and thus the compact perturbation results which we rely on in the proofs cannot be applied.
\end{remark}

\section{Conclusions}
\label{sec:conclusions}
We have extended the locality results of \cite{ChenOrtnerThomas2019:locality} to non-linear tight binding models. More specifically, the results of this paper are twofold: \textit{(i)} we have written analytic quantities of interest (which includes the total energy of the system) as the sum of exponentially localised site contributions. Moreover, \textit{(ii)} under a mild assumption on the electronic density, we have shown that point defects in the material only weakly affect the locality estimates. That is, away from the defect, where the local atomic environment resembles that of the corresponding defect-free configuration, the locality estimates resemble that of the defect-free case.     

The results of this paper represent a first natural stepping stone between the linear tight binding results of \cite{ChenOrtner16, ChenLuOrtner18,ChenOrtnerThomas2019:locality} towards more accurate electronic structure models, such as Kohn-Sham density functional theory. 

As well as justifying a number of interatomic potential and multi-scale methods
\cite
    {ChenOrtner16,
    ChenOrtner17,
    CsanyiAlbaretMorasPayneDeVita05},
we may use the locality results of this paper to formulate limiting variational problems for infinite systems. That is, for a fixed configuration $u_0$ with associated self-consistent electronic density $\rho_0$ such that $(u_0, \rho_0)$ is stable, we can renormalise the total energy and define 
\begin{align}\label{eq:grand_pot_diff_functional}
    \mathcal G^\beta(u) \coloneqq 
        \sum_{\ell\in\Lambda} \left( 
                                G^\beta_\ell(u) - G^\beta_\ell(u_0) 
                            \right)
\end{align}
where $G^\beta_\ell$ is given by \cref{eq:local_integral_sc} with 
$\mathfrak{g}^\beta(z) = \frac{2}{\beta} \log( 1 - f_\beta(z-\mu) )$ 
for finite Fermi-temperature and $\mathfrak{g}^\infty(z) = 2(z - \mu) \chi_{(-\infty,\mu)}(z)$ in the case of zero Fermi-temperature. By the stability of the configuration $(u_0,\rho_0)$, it follows from the locality results of this paper together with \cite{ChenNazarOrtner19} that \cref{eq:grand_pot_diff_functional} is well defined in a $\|D\cdot\|_{\ell^2_\ctGamma(\Lambda)}$-neighbourhood of $u_0$.

We can then consider the following geometry relaxation problems
\begin{align}\label{eq:min_problem}
    \overline{u} \in \arg\min
        \left\{ 
            \mathcal G^\beta(u) 
            \colon u \in B_\delta(u_0;\|D\cdot\|_{\ell^2_\ctGamma}) 
            \text{ satisfies \textbf{(L)}}
        \right\}
\end{align}
where ``$\arg\min$'' denotes the set of local minimisers. We emphasise here that, in order to define these problems, we require the differentiability of the site energies and so can only define these problems locally around stable configurations. We may follow the proofs of \cite{OrtnerThomas2020:pointdef}, to extend the results to the case of non-linear tight binding models. For example, we may show the following:
\begin{theorem}\label{thm:limits}
    Suppose that $\mu$ is fixed such that \asGAP~is satisfied and that $\overline{u}$ solves \cref{eq:min_problem} for $\beta = \infty$ such that  
    $\Braket
        {\delta^2\mathcal{G}^\infty(\overline{u})v,
        v} 
    \geq c_0\| Dv \|^2_{\ell^2_\ctGamma}$
    for all $v \in \dot{\mathscr W}^{1,2}(\Lambda)$ and some $c_0>0$. Then, there exists $\overline{u}_\beta$ solving \cref{eq:min_problem} with $\beta<\infty$ such that
    $\|D(\overline{u}_\beta - \overline{u})\|_{\ell^2_\ctGamma} \lesssim e^{-c\beta}$. 
\end{theorem}
\begin{proof}[Sketch of the Proof]
    Here, to distinguish between the finite and zero Fermi-temperature cases, we will write $F^{\beta}(u;\rho)$, $\mathscr L^\beta(u;\rho)$ and $F^\infty(u;\rho)$, $\mathscr L^\infty(u;\rho)$, respectively.

    Firstly, we note that there exists a locally unique electronic density $\overline{\rho} = F^\infty(\overline{u};\overline{\rho})$. By stability, $I - \mathscr L^\infty(\overline{u};\overline{\rho})$ is invertible and so it follows from zero Fermi-temperature limit results (see \cite[Lemma~5.9]{OrtnerThomas2020:pointdef}) that
    $I - \mathscr L^\beta(\overline{u};\overline{\rho}) = \left( I - \mathscr L^\infty(\overline{u};\overline{\rho}) \right) - (\mathscr L^\beta(\overline{u};\overline{\rho}) - \mathscr L^\infty(\overline{u};\overline{\rho}))$
    is also invertible for all sufficiently large $\beta$. This means that for ${u}$ in a neighbourhood of $\overline{u}$, there exists a locally unique electronic density ${\rho}_\beta$ satisfying ${\rho}_\beta = F^\beta({u};{\rho}_\beta)$. Therefore, for ${u}$ in a neighbourhood of $\overline{u}$, we may write $G^\beta_\ell({u}) \coloneqq G^\beta_\ell({u};{\rho}_\beta)$. This means that, for $\beta$ sufficiently large, we may apply the inverse function theorem on $\delta\mathcal G^\beta$ about $\overline{u}$ as in \cite[Theorem~2.3]{OrtnerThomas2020:pointdef}.
\end{proof}

In Theorem~\ref{thm:limits}, we restrict ourselves to the grand-canonical ensemble where there is a fixed chemical potential. By following the proofs of \cite{OrtnerThomas2020:pointdef}, one can also show analogous results for the canonical ensemble where the Fermi-level arises as a Lagrange multiplier for the particle number constraint.  

We believe that the thermodynamic limit results of \cite{OrtnerThomas2020:pointdef} can also be extended to the setting of non-linear tight binding models. The only additional technical detail is to show that the limiting configuration gives rise to stable configurations defined along the sequence of finite domain approximations. This means the choice of boundary condition and the number of electrons imposed plays a key role in the analysis. While we do not see any problem in extending the results of \cite{OrtnerThomas2020:pointdef} for a supercell model, it is much less clear how and when the boundary effects may inhibit the stability of the electronic structure when considering clamped boundary conditions, for example.

\section*{Acknowledgements}
Helpful discussions with Huajie Chen, Genevi\`{e}ve Dusson, Antoine Levitt and Christoph Ortner are gratefully acknowledged. The content of \S\ref{sec:results_general} is an extension of Antoine Meyer’s MMath thesis \cite[\S3.2]{AntoineMasters}.

\section{Proofs: General Locality Estimates}
\label{sec:proofs}
In order to simplify the notation in the following, we will write 
$\Ham(u;\rho) = \HamLin(u) + \HamNonLin(\rho)$
where 
    $\HamLin(u)_{\ell k}^{ab} 
        \coloneqq h_{\ell k}^{ab}( \bm{r}_{\ell k}(u) )$
and 
    $\HamNonLin(\rho)_{\ell k}^{ab}
        = v(\rho(\ell)) \delta_{\ell k} \delta_{ab}$.
Further, we write denote the resolvent operator by $\mathscr R_z(u;\rho) \coloneqq (\Ham(u;\rho) - z)^{-1}$.

                    \subsection{Preliminaries}                        

Firstly, we prove that we may write the local quantities of interest as a function of the displacement, \cref{eq:local_integral_sc}:
\begin{lemma}\label{lem:perturbation}
    Suppose that $(u,\rho)$ satisfies \asNonInter, \asSC~and \asStab. Then, there exist $\delta_u, \delta_\rho > 0$ such that for all 
        $\widetilde{u} \colon \Lambda \to \mathbb R^d$ 
        with $\|D(\widetilde{u} - u)\|_{\ell^2_\ctGamma(\Lambda)} < \delta_u$,
    there exists a unique electronic density 
    $\widetilde{\rho} = \widetilde{\rho}(\widetilde{u})$ 
    satisfying 
    $\| \widetilde{\rho} - \rho\|_{\ell^2(\Lambda)} < \delta_\rho$
    and 
    $\widetilde\rho = F(\widetilde u;\widetilde\rho)$. 
    Further, the mapping $\widetilde{u} \mapsto \widetilde{\rho}$ is smooth.
\end{lemma}
\begin{proof}[Sketch of the Proof]
    We apply the implicit function theorem on 
    $T(\widetilde{u};\widetilde \rho) \coloneqq \widetilde\rho - F(\widetilde u;\widetilde \rho)$, 
    a smooth map in a neighbourhood of $(u,\rho)$. Since $\rho$ satisfies \asSC, we have $T(u;\rho) = 0$. By \asStab, we have: for each $\widetilde{u}$ in a neighbourhood of $u$, there exists a locally unique $\widetilde \rho = \widetilde\rho(\widetilde u)$ with $\widetilde\rho = F(\widetilde u;\widetilde\rho)$. 

    The fact that $F$ is indeed a smooth map in a neighbourhood of $(u,\rho)$ follows from the fact that small perturbations in $u$ and $\rho$ lead to small perturbations in $\sigma(\Ham(u;\rho))$ \cite{Kato95}. This means that the fixed contour $\mathscr C_f$, which depends on $(u,\rho)$, can be used in the definition of $F$ in a neighbourhood of $(u,\rho)$. 
    
    For full details in a slightly different setting, see \cite[Theorem~5.3]{ELu10}.
\end{proof}

We now state a Combes-Thomas type estimate \cite{CombesThomas1973} for the resolvent:

\begin{lemma}[Combes-Thomas Resolvent Estimates]
\label{lem:CT}
    Given $u \in \dot{\mathscr{W}}^{1,2}(\Lambda)$ satisfying \asNonInter, suppose that $T(u)$ is an operator on $\ell^2(\Lambda\times\{1,\dots,\numorbitals\})$ given by 
    \[
        [T(u)w](\ell;a) \coloneqq \sum_{k \in \Lambda}\sum_{1\leq b\leq \numorbitals} T(u)_{\ell k}^{ab}w(k;b)
        \quad \textup{where} \quad 
        |T(u)_{\ell k}^{ab}|\leq c_T e^{-\gamma_T r_{\ell k}(u)} 
    \]
    for some $c_T,\gamma_T > 0$. Then, if $z \in \mathbb C$ with $\mathfrak{d} \coloneqq \mathrm{dist}(z,\sigma(T(u))) > 0$, we have,
    \begin{align*}
        \left|\left[ (T(u) - z)^{-1} \right]_{\ell k}^{ab} \right| 
        \leq 2 \mathfrak{d}^{-1} e^{-\ctCT r_{\ell k}(u)}
    \end{align*}
    where 
    $\ctCT \coloneqq c\gamma_T \min\{1, c_T^{-1}\gamma_T^d \mathfrak{d}\}$
    where $c>0$ depends only on $\|Du\|_{\ell^2_\ctGamma}$, $\mathfrak{m}$ and $d$.
\end{lemma}
\begin{proof}
The proof is analogous to \cite[Lemma~6]{ChenOrtner16} where the claimed $\mathfrak{d}$-dependence can be obtained by following the same proof and calculating the pre-factor in \cite[(4.4)]{ChenOrtnerThomas2019:locality}.

That is, it can be shown that 
\begin{align}\label{CT:proof}\begin{split}
    \sup_{\ell \in \Lambda} \sum_{k \in \Lambda} c_T e^{-\gamma_T r_{\ell k}(u)}
    (e^{\gamma_\mathrm{CT} r_{\ell k}(u)} - 1) \leq C \frac{c_T}{\gamma_T^{d+1}} \gamma_\mathrm{CT}, \qquad \forall\,\gamma_\mathrm{CT} \leq \frac{1}{2}\gamma_T
\end{split}\end{align}
for some $C>0$ depending only on $\|Du\|_{\ell^2_\ctGamma}$, $\mathfrak{m}$ and $d$. The proof follows by choosing $\gamma_\mathrm{CT} > 0$ sufficiently small such that the right hand side of \cref{CT:proof} is less than $\frac{1}{2}\mathfrak{d}$.
\end{proof}

\begin{remark}
    More careful analysis reveals that the above proof gives
    \[
        \ctCT \coloneqq \frac{1}{2}\gamma_T 
        \min\left\{
            1, 
            \left(
                \frac{\mathfrak{m}}{2}
            \right)^{d+1} \frac
                            {\gamma_T^{d} \mathfrak{d}}
                            {d! \|Du\|_{\ell^\infty} c_T}
        \right\}
        \quad \textup{where} \quad 
        \|Du\|_{\ell^\infty} \coloneqq 
            \sup_{\ell \in \Lambda} 
            \sup_{\rho \in \Lambda - \ell}
            \frac
                {|D_{\rho}u(\ell)|}
                {|\rho|}.
    \]
    Here, we have used the fact that $\|D\cdot\|_{\ell^\infty}$ defines a semi-norm that is equivalent to $\|D\cdot\|_{\ell^2_\ctGamma}$ \cite{ChenNazarOrtner19}.
\end{remark}

By applying Lemma~\ref{lem:CT} to $\Ham(u;\rho)$, we obtain locality estimates for the resolvents $\mathscr R_z(u;\rho)$: for $z \in \mathbb C$ with 
$\mathrm{dist}  \left(
                    z,\sigma( \Ham(u;\rho) )
                \right) \geq \mathfrak{d} > 0$
we have
\begin{gather}
    \left| \mathscr{R}_z(u;\rho)_{\ell k}^{ab} \right| 
    \leq 2 \mathfrak{d}^{-1} e^{-\gamma_\mathrm{r}(\mathfrak{d}) r_{\ell k}(u) }         \label{eq:resolvent_bound}
\end{gather}
where 
$\gamma_\mathrm{r}(\mathfrak{d}) \coloneqq c \min\{1,\mathfrak{d}\}$
and $c$ is a positive constant that depends only on 
$h_0$, $\gamma_0$, $\|Du\|_{\ell^2_\ctGamma}$, $\mathfrak{m}$ and $d$.
We will apply \cref{eq:resolvent_bound} for both 
$z \in \mathscr C_f$ (with $\mathfrak{d} = \mathsf{d}_f$)
and 
$z \in \mathscr C_\mathfrak{o}$ (with $\mathfrak{d} = \mathsf{d}_\mathfrak{o}$).

Therefore, by \cref{eq:L}, we have
\begin{gather}
    \left| \mathscr{L}(u;\rho)_{\ell k}^{ab} \right| 
    \leq C \mathsf{d}_{f}^{-2}
    e^{-2\gamma_\mathrm{r}(\mathsf{d}_f) r_{\ell k}(u) }.
                \label{eq:L_bound}
\end{gather}
Therefore, applying Lemma~\ref{lem:CT} again with $T$ replaced with $\mathscr{L}(u;\rho)$ (and with $z = 1$), we obtain, 
\begin{align}
\label{eq:L_resolvent_bound}
    \left|
        \left[
            (I - \mathscr L(u;\rho))^{-1}
        \right]_{\ell k}^{ab} 
    \right|
    \leq 2 \mathsf{d}_{\mathscr{L}}^{-1} e^{-\gamma_\mathrm{s} r_{\ell k}(u) }
\end{align}
where $\mathsf{d}_\mathscr{L}$ is the constant from \cref{eq:dstab} and 
$\gamma_\mathrm{s} \coloneqq c_1 \gamma_\mathrm{r}(\mathsf{d}_f) \min\{1,\mathsf{d}_{f}^2\gamma_\mathrm{r}^d \mathsf{d}_{\mathscr L}\}$.
By expanding $\gamma_\mathrm{r}$ in terms of $\mathsf{d}_f$, we obtain 
$\gamma_\mathrm{s} = 
c_2 \min\left\{
            1, 
            \mathsf{d}_{f}, 
            \mathsf{d}_{f}^2\mathsf{d}_{\mathscr L}, \mathsf{d}_{f}^{d+3}\mathsf{d}_{\mathscr L}
        \right\}$
for some $c_2 >0$ depending only on $h_0, \gamma_0, \|Du\|_{\ell^2_\ctGamma}, \mathfrak{m}, d$, the length of $\mathscr C_f$, $\|v^\prime\|_\infty$ and $\numorbitals$. 

\subsection{Proof of Theorem~\ref{thm:locality}: General Locality Estimates}

We are now in a position to prove the locality estimates. Since we may write $O^\mathrm{sc}_\ell$ as an integral of the resolvent operator, derivatives of $O^\mathrm{sc}_\ell$ can be written as derivatives of the resolvent operators. 

We start with the case $j = 1$: for $z \in \mathbb C$ with $\mathrm{dist}(z,\sigma(\Ham(u;\rho))) \geq \mathfrak{d}>0$, we have
\begin{align}
\label{eq:derivative_resolvent}
\begin{split}
    &\frac
        { \partial[\mathscr R_z(u;\rho)]^{aa}_{\ell\ell} }
        { \partial[u(m)]_i }
    = - \left[ 
            \mathscr R_z(u;\rho) 
            \frac
                {\partial \left[ \Ham(u;\rho) \right]}
                {\partial [u(m)]_i}
            \mathscr R_z(u;\rho) 
        \right]^{aa}_{\ell\ell} \\
    &\qquad= - \left[ 
                    \mathscr R_z(u;\rho) 
                    \frac
                        {\partial \left[\HamLin(u;\rho)\right]}
                        {\partial [u(m)]_i}
                    \mathscr R_z(u;\rho) 
                \right]^{aa}_{\ell\ell} 
    - \sum_{k\in \Lambda} \sum_{b=1}^{\numorbitals} 
    \left(
        [ \mathscr R_z(u;\rho) ]^{ab}_{\ell k}
    \right)^2 
    v^\prime(\rho(k)) 
    \frac
        {\partial \rho(k)}
        {\partial [u(m)]_i}. 
\end{split}
\end{align}
Here, we have used the fact that,
\begin{align*}
    \frac
        {\partial   \left[
                        \Ham(u;\rho)
                    \right]_{\ell k}^{ab}}
        {\partial [u(m)]_i} 
    &= \frac
            {\partial \HamLin(u)_{\ell k}^{ab}}
            {\partial [u(m)]_i}
      + v^\prime(\rho(\ell)) 
        \frac
            {\partial \rho(\ell)}
            {\partial [u(m)]_i}
        \delta_{\ell k}\delta_{ab}.
\end{align*}
The first contribution in \cref{eq:derivative_resolvent} can be treated by applying \cref{eq:resolvent_bound} as in \cite{ChenOrtner16,ChenOrtnerThomas2019:locality}:
\begin{align}
    \label{eq:derivative_resolvent_linear}
    \begin{split}
        \left|
            \left[ 
                \mathscr R_z(u;\rho) 
                \frac
                    {\partial [\HamLin(u)]}
                    {\partial [u(m)]_i} 
                \mathscr R_z(u;\rho) 
            \right]_{\ell\ell}^{aa}
        \right| 
        %
        %
        %
        %
        &\leq C\mathfrak{d}^{-2} 
                e^{-\min\{ 
                        \gamma_\mathrm{r}(\mathfrak{d}), 
                        \gamma_0 
                    \} r_{\ell m}(u)}.
    \end{split}
\end{align}

Now we move on to consider the non-linear contribution in \cref{eq:derivative_resolvent}. By taking derivatives in the self-consistency equation for $\rho$ (that is, $\rho = F(u;\rho)$ from \asSC), we obtain the following identity,
\begin{align*}
    \frac
        {\partial \rho(\ell)}
        {\partial [u(m)]_i} 
    &= - \frac
            {1}
            {2\pi i} 
        \sum_{a} \oint_{\mathscr{C}_f} f(z-\mu)
        \frac
            { \partial[\mathscr R_z(u;\rho)]_{\ell\ell} }
            { \partial[u(m)]_i } \mathrm{d}z\\ 
    &= \frac
                {1}
                {2\pi i} \sum_a \oint_{\mathscr{C}_f} f(z-\mu) 
            \left[ 
                \mathscr R_z(u;\rho) 
                \frac
                    {\partial \left[\HamLin(u)\right]}
                    {\partial [u(m)]_i}
                \mathscr R_z(u;\rho) 
            \right]^{aa}_{\ell\ell} \mathrm{d}z 
        + \left[
            \mathscr L(u;\rho) \frac
                                    {\partial \rho}
                                    {\partial [u(m)]_i}
        \right]_{\ell}
\end{align*}
where $\mathscr L(u;\rho)$ is the stability operator given in \cref{eq:L}. That is, 
\begin{align}\label{eq:d_rho_u}
\frac
    {\partial \rho(\ell)}
    {\partial [u(m)]_i} 
    =   \left[
            (I - \mathscr L(u;\rho))^{-1} \phi^{(m)} 
        \right]_\ell
\end{align}
where $\phi^{(m)} \in \ell^2(\Lambda)$ is given by 
\begin{align}\label{eq:phi^m}
    \phi^{(m)}_\ell \coloneqq  
        \frac{1}{2\pi i} \sum_{a = 1}^{\numorbitals} 
        \oint_{\mathscr{C}_f} f(z-\mu) 
        \left[ \mathscr R_z(u;\rho) 
            \frac
                {\partial \left[ \HamLin(u;\rho) \right]}
                {\partial [u(m)]_i} 
            \mathscr R_z(u;\rho) 
        \right]^{aa}_{\ell\ell}
        \mathrm{d}z.
\end{align}
Applying \cref{eq:derivative_resolvent_linear} and using the fact that $f$ is uniformly bounded, we have 
\begin{align}
\label{eq:phi^m_bound}
    \big| \phi^{(m)}_\ell \big| 
    \leq C \mathsf{d}_{f}^{-2} 
        e^{-\min\{
                \gamma_\mathrm{r},
                \gamma_0
            \} r_{\ell m}(u)}.
\end{align}
Combining \cref{eq:phi^m_bound} with the resolvent estimate for 
$(I - \mathscr L(u;\rho))^{-1}$
from \cref{eq:L_resolvent_bound}, we obtain 
\begin{align}\label{eq:d_rho_u_bound}
\begin{split}
    \left| \frac
                {\partial \rho(\ell)}
                {\partial [u(m)]_i}
    \right| 
    &\leq C
    \mathsf{d}_{f}^{-2} 
    \mathsf{d}_{\mathscr L}^{-1} 
    \sum_{k \in \Lambda} e^{-\gamma_\mathrm{s} r_{\ell k}(u)} 
    e^{-\min\{
                \gamma_\mathrm{r},
                \gamma_0
            \} r_{km}(u)} \\
    &\leq C
    \mathsf{d}_{f}^{-2} 
    \mathsf{d}_{\mathscr L}^{-1} 
    e^{-\frac{1}{2}\min\{
                        \gamma_\mathrm{s}, 
                        \gamma_\mathrm{r}, 
                        \gamma_0
                        \} r_{\ell m}(u) }.
\end{split} 
\end{align}
Therefore, we may bound the second term in \cref{eq:derivative_resolvent}:  for $z \in \mathbb C$ with $\mathrm{dist}(z,\sigma(\Ham(u;\rho))) \geq \mathfrak{d}$, we have
\begin{align}
\label{eq:derivative_resolvent_non-linear}
\begin{split}
    &\sum_{k \in \Lambda}
    \sum_{1\leq b \leq \numorbitals}
    \left(
        \mathscr{R}_z(u;\rho)_{\ell k}^{ab}
    \right)^2 
    v^\prime(\rho(k)) 
    \frac
        {\partial \rho(k)}
        {\partial[u(m)]_i} \\
    &\qquad\leq C \mathfrak{d}^{-2}\mathsf{d}_{f}^{-2} 
        \mathsf{d}_{\mathscr L}^{-1} \sum_{k \in \Lambda} |v^\prime(\rho(k))|
    e^{-2\gamma_\mathrm{r}(\mathfrak{d}) r_{\ell k}(u)}
    e^{-\frac{1}{2}\min\{
                            \gamma_\mathrm{s}, 
                            \gamma_\mathrm{r}(\mathsf{d}_f),
                            \gamma_0\} r_{k m}(u)} \\
    &\qquad\leq C \mathfrak{d}^{-2}\mathsf{d}_{f}^{-2} 
        \mathsf{d}_{\mathscr L}^{-1} 
        \|v^\prime\|_{\infty} 
        e^{-\frac{1}{4}\min\{
                            \gamma_\mathrm{s},
                            4\gamma_\mathrm{r}(\mathfrak{d}),
                            \gamma_\mathrm{r}(\mathsf{d}_f),
                            \gamma_0 
                        \} r_{\ell m}(u)}.
\end{split}
\end{align}
Combining \cref{eq:derivative_resolvent_linear} and \cref{eq:derivative_resolvent_non-linear} with $\mathfrak{d} = \mathsf{d}_\mathfrak{o}$ and using the fact that $\mathfrak{o}$ is uniformly bounded along the contour $\mathscr C_\mathfrak{o}$, we can conclude the proof for $j=1$.

Higher derivatives can be treated by taking derivatives of \cref{eq:derivative_resolvent}. The first contribution in \cref{eq:derivative_resolvent} is what arises in the linear case and so derivatives of this term can be treated in the same way as in \cite{ChenOrtnerThomas2019:locality}. We sketch the argument for $j=2$ for the second contribution in \cref{eq:derivative_resolvent}. We fix $k \in \Lambda$ and $b\in \{1,\dots,\numorbitals\}$ and note
\begin{align}\label{eq:2nd_derivative}
\begin{split}
    &\frac
        {\partial }
        {\partial u(n)}
        \left\{
            \left(
                [\mathscr R_z(u;\rho)]^{ab}_{\ell k}
            \right)^2 v^\prime(\rho(k)) 
            \frac
                {\partial \rho(k)}
                {\partial u(m)} 
        \right\} 
    = 2 [\mathscr R_z(u;\rho)]^{ab}_{\ell k} 
        \frac
            {\partial [\mathscr R_z(u;\rho)]^{ab}_{\ell k}}
            {\partial u(n)} 
        v^\prime(\rho(k)) \frac
                            {\partial \rho(k)}
                            {\partial u(m)} \\
    &\qquad+ \left(
                [\mathscr R_z(u;\rho)]^{ab}_{\ell k}
            \right)^2 v^{\prime\prime}(\rho(k)) 
            \frac
                {\partial \rho(k)}
                {\partial u(n)}
            \frac
                {\partial \rho(k)}
                {\partial u(m)}
    +   \left(
            [\mathscr R_z(u;\rho)]^{ab}_{\ell k}
        \right)^2 v^{\prime}(\rho(k)) 
        \frac
            {\partial^2 \rho(k)}
            {\partial u(n) \partial u(m)}.
\end{split}\end{align}
After summing over $k \in \Lambda$, the first two contributions in \cref{eq:2nd_derivative} may be bounded above by a constant multiple of 
$e^{-\eta (r_{\ell m}(u) + r_{\ell n}(u))}$ 
for some $\eta > 0$ depending only on the exponents in \cref{eq:resolvent_bound}, \cref{eq:derivative_resolvent_linear}, and \cref{eq:derivative_resolvent_non-linear}. The final contribution in \cref{eq:2nd_derivative} involves the second derivative of the electronic denisty which may be bounded above as follows: using \cref{eq:d_rho_u}, we have
\begin{align}
\label{eq:2nd_derivative_rho}
\begin{split}
    \bigg|
        \frac
            {\partial^2 \rho(k)}
            {\partial u(n) \partial u(m)}
    \bigg|
    &\leq\bigg|
            \bigg[
                \frac
                    {\partial (I - \mathscr L(u;\rho))^{-1}}
                    {\partial u(n)} \phi^{(m)}
            \bigg]_k
        \bigg|
    + \bigg|
        \bigg[
            (I - \mathscr L(u;\rho))^{-1} 
            \frac
                {\partial \phi^{(m)}}
                {\partial u(n)}
        \bigg]_k
    \bigg| \\
    &\leq C e^{-\eta (r_{kn}(u) + r_{km}(u)) }
\end{split}
\end{align}
where $\eta > 0$ depends only on the exponents in \cref{eq:resolvent_bound}, \cref{eq:L_resolvent_bound}, \cref{eq:d_rho_u_bound} and in the locality estimates of the first contribution in \cref{eq:derivative_resolvent}. The estimate in \cref{eq:2nd_derivative_rho} is easy to prove but is lengthy and very similar to the calculations above and so is omitted. Using \cref{eq:2nd_derivative_rho} and summing over $k \in \Lambda$ in \cref{eq:2nd_derivative} we can conclude.

\section{Proofs: Improved Locality Estimates}
\label{sec:proofs_improved}
Before we prove Theorem~\ref{thm:improved_locality}, we require an improved Combes-Thomas type estimate for the resolvent operators; see Lemma~\ref{lem:CT-improved}, below. In the following section, we discuss this result and explain how we can use it despite the fact that the reference and defect Hamiltonians are defined on different spatial domains. Then, in \S\ref{sec:prooftheorem4}, we show that the operators 
$\Ham(u;\rho)$ and $\mathscr L(u;\rho)$
satisfy the conditions of Lemma~\ref{lem:CT-improved} and thus prove Theorem~\ref{thm:improved_locality}.

                    \subsection{Preliminaries}                  

We show an improved resolvent estimate for operators on $\ell^2(\Lambda\times\{1,\dots,\numorbitals\})$ that can be decomposed into a reference operator and two perturbations that are small in the sense of rank and Frobenius norm (Lemma~\ref{lem:CT-improved}), respectively. First, we need a basic identity for the inverse of an updated operator:
\begin{lemma}[Woodbury \cite{Hager1989}]
\label{lem:Woodbury}
    Suppose that $A$ and $P$ are operators on a Banach space such that $A$ and $A + P$ are invertible. Then, $I + PA^{-1}$ and $I + A^{-1}P$ are invertible and 
    \begin{align*}
        (A + P)^{-1} 
        &= A^{-1} - A^{-1} (I + PA^{-1})^{-1} P A^{-1} \\
        &= A^{-1} - A^{-1}P (I + A^{-1}P)^{-1} A^{-1}.
    \end{align*}
\end{lemma}
\begin{proof}
    Firstly, $I + PA^{-1} = (A + P)A^{-1}$ is invertible with inverse $A(A+P)^{-1}$. Therefore, we have
    $
        (A+P)^{-1} - A^{-1} 
            = (A+P)^{-1}[A - (A + P)] A^{-1} 
            = - [(I+PA^{-1})A]^{-1}PA^{-1}
    $.
    The second formulation can be shown similarly.
\end{proof}

Using this Woodbury identity, we may prove the following ``improved'' Combes-Thomas estimate:
\begin{lemma}[Improved Combes-Thomas Resolvent Estimate]\label{lem:CT-improved}
    Suppose 
    $\delta, R > 0$, $T^\mathrm{ref}, T^\mathrm{FR}, T^\delta$
    are operators on 
    $\ell^2(\Lambda \times \{1,\dots,\numorbitals\})$
    and define 
    $T \coloneqq T^\mathrm{ref} + T^\mathrm{FR} + T^\delta$.
    Further, suppose that:
    \begin{itemize}
    \itemsep0em 
        \item  $\left|
                    [ T^\mathrm{ref} ]_{\ell k}^{ab}
                \right| 
                \leq c_T e^{-\gamma_T |\ell - k|}$ 
                for some $c_T, \gamma_T > 0$,
        
        \item   $[T^\mathrm{FR}]_{\ell k}^{ab} = 0$ 
        if $\ell \not\in \Lambda \cap B_{R}$ or $k \not\in \Lambda \cap B_{R}$,
        
        \item $\|T^\delta\|_\mathrm{F} \leq \delta$ and,
        
        \item $z \in \mathbb C$ with 
        $\mathfrak{d} \coloneqq \mathrm{dist}(z,\sigma(T)) > 0$ 
        and 
        $\mathfrak{d}^\mathrm{ref} \coloneqq \mathrm{dist}(z,\sigma(T^\mathrm{ref})) - \delta > 0$.
    \end{itemize}
    Then, there exists a constant $C$, depending on 
    $\delta$, $R$, $c_T$, $\gamma_T$, $\|T^\mathrm{FR}\|_\infty$, $\mathfrak{d}$, $\mathfrak{d}^\mathrm{ref}$ and $d$,
    such that
    \begin{gather*}
        \left|
            \left[
                (T - z)^{-1}
            \right]_{\ell k}^{ab}
        \right| 
        \leq C_{\ell k} e^{-\ctCT(\mathfrak{d}^\mathrm{ref}) |\ell - k|} 
        \quad \textup{where}\\
        C_{\ell k} \coloneqq 2 (\mathfrak{d}^\mathrm{ref})^{-1} 
            + C(1+|z|)^{2} 
            e^{-\ctCT(\mathfrak{d}^\mathrm{ref}) 
            \left(|\ell| + |k| - |\ell - k|\right)},
    \end{gather*}
    and 
    $\ctCT(\mathfrak{d}^\mathrm{ref}) \coloneqq 
        c \gamma_T \min\{ 
                        1, 
                        c_T^{-1}(\gamma_T)^d \mathfrak{d}^\mathrm{ref} 
                    \}$
    is the constant from Lemma~\ref{lem:CT}.
\end{lemma}

\begin{proof}
This proof closely follows the ideas of \cite[\S4.4]{ChenOrtnerThomas2019:locality} but for more general operators $T$. We sketch the argument for completeness.

After defining 
$\mathscr R_z \coloneqq (T - z)^{-1}$
and 
$\mathscr R_z^\delta \coloneqq (T^\mathrm{ref} + T^\delta - z)^{-1}$,
we apply Lemma~\ref{lem:Woodbury} to obtain:
\begin{align}
\label{eq:woodbury}
\begin{split}
    \mathscr{R}_z &= 
        \mathscr{R}^\delta_z 
        - \mathscr{R}^\delta_z 
            (I + T^\mathrm{FR}\mathscr{R}^\delta_z)^{-1} 
            T^\mathrm{FR} \mathscr{R}^\delta_z 
    =   \mathscr{R}^\delta_z 
        - \mathscr{R}^\delta_z T^{\mathrm{FR}}
            (I + \mathscr{R}^\delta_zT^\mathrm{FR})^{-1}
            \mathscr{R}^\delta_z.
\end{split}
\end{align}
We will consider the two terms in \cref{eq:woodbury} separately. 

Firstly, since 
$\mathrm{dist}\left(\sigma(T^\mathrm{ref}), \sigma(T^\mathrm{ref} + T^\delta)\right) \leq \|T^\delta\|_\mathrm{F} \leq \delta$
\cite{Kato95}, we may apply Lemma~\ref{lem:CT} directly to conclude that 
$
| [\mathscr R_z^\delta ]_{\ell k}^{ab} | 
    \leq 2 (\mathfrak{d}^{\mathrm{ref}})^{-1} 
        e^{-\ctCT(\mathfrak{d}^\mathrm{ref}) |\ell - k| }
$.

Next, we note that 
$(I + T^\mathrm{FR}\mathscr{R}^\delta_z)^{-1}T^\mathrm{FR}$ 
is a finite rank operator with 
$\big[
    (I + T^\mathrm{FR}\mathscr{R}^\delta_z)^{-1}
    T^\mathrm{FR}
\big]_{\ell k}^{ab} 
=
\big[
    T^\mathrm{FR}
    (I + \mathscr{R}^\delta_zT^\mathrm{FR})^{-1}
\big]_{\ell k}^{ab} 
= 0$
for all 
$(\ell,k) \not \in (\Lambda \cap B_R)^2$. 
Therefore,
\begin{align}\label{eq:resolvent_extra_bit}
    \left|
        [
            \mathscr{R}^\delta_z 
            (I + T^\mathrm{FR}\mathscr{R}^\delta_z)^{-1}
            T^\mathrm{FR}\mathscr{R}^\delta_z
        ]_{\ell k}^{ab}
    \right| 
    &\leq\|
            (I + T^\mathrm{FR}\mathscr{R}^\delta_z)^{-1}
            T^\mathrm{FR}
        \|_\mathrm{max}
        \sum_{\ell_1,\ell_2 \in \Lambda \cap B_R} 
        e^{-\ctCT(\mathfrak{d}^\mathrm{ref})
            \left( 
                |\ell-\ell_1| + |\ell_2 - k|
            \right)} 
                                            \nonumber\\
    &\leq C \|
                (I + T^\mathrm{FR}\mathscr{R}^\delta_z)^{-1}
                T^\mathrm{FR}
            \|_\mathrm{max} 
            e^{-\ctCT(\mathfrak{d}^\mathrm{ref})
                \left( 
                    |\ell| + |k|
                \right)}.
\end{align}

We obtain the claimed $z$-dependence after showing that 
$
\|
    (I + T^\mathrm{FR}\mathscr{R}^\delta_z)^{-1}
    T^\mathrm{FR}
\|_\mathrm{max} 
\lesssim (1+|z|)^2
$. 
By \cref{eq:woodbury}, we have
\begin{align}\label{eq:resolvent_extra}
\begin{split}
    (I + T^\mathrm{FR}\mathscr{R}^\delta_z)^{-1}
    T^\mathrm{FR}
    &=  (T^\mathrm{ref} + T^\delta - z)
        (\mathscr{R}_z - \mathscr{R}_z^\delta)
        (T^\mathrm{ref} + T^\delta - z) \\
    &= -(T^\mathrm{ref} + T^\delta - z)
        \mathscr{R}_z T^\mathrm{FR} \mathscr{R}_z^\delta
        (T^\mathrm{ref} + T^\delta - z).
\end{split}
\end{align}
Using the fact 
$[T^{\mathrm{FR}}]_{\ell_1 \ell_2}^{ab} = 0$
unless $\ell_1, \ell_2 \in \Lambda \cap B_R$, we have
\begin{align}\label{eq:resolvent_extra_2}
    |[
        \mathscr{R}_z T^\mathrm{FR} \mathscr{R}_z^\delta
    ]_{\ell k}^{ab}|
    \leq C  \sum_{\ell_1,\ell_2 \in \Lambda \cap B_R}
            e^{-\ctCT(\mathfrak{d}^\mathrm{ref})|\ell - \ell_1|} 
            e^{-\ctCT(\mathfrak{d}) |\ell_2 - k|} 
    \leq C e^{-\eta (|\ell| + |k| - 2R)}
\end{align}
where 
$
\eta \coloneqq \frac{1}{2}
    \min\{ 
        \ctCT(\mathfrak{d}), 
        \ctCT(\mathfrak{d}^\mathrm{ref}) 
    \}
$. 
Therefore, by combining \cref{eq:resolvent_extra} and \cref{eq:resolvent_extra_2}, we obtain
\begin{align*}
    \left|
        \big[
            (I+T^\mathrm{FR}\mathscr R_z^\delta)^{-1}
            T^\mathrm{FR}
        \big]_{\ell k}^{ab}
    \right| 
    &\leq C \bigg( 
                \sum_{\ell_1 \in \Lambda} 
                \left( 
                    c_T + \delta + |z|
                \right) e^{-\eta(|\ell_1| - R)} 
            \bigg)^2
    \leq C(1 + |z|)^2
\end{align*}
and can thus conclude.
\end{proof}

We will now show that $\Ham(u;\rho)$ and $\mathscr L(u;\rho)$ can be written as in the statement of Lemma~\ref{lem:CT-improved} so that we may apply these improved resolvent estimates. 

\label{pg:spatial_domains}Since $\Ham(u;\rho)$ and $\Ham^\mathrm{ref}$ are defined on different spatial domains, we cannot directly compare the Hamiltonian with the corresponding reference operator. In order to alleviate this issue, we follow the arguments of \cite{ChenOrtnerThomas2019:locality}. Firstly, we shift the operator by a constant multiple of the identity $c I$ and replace the contour and chemical potential by 
$\mathscr C_\mathfrak{o} + c$ and $\mu + c$, 
respectively, so that $0$ is not encircled by $\mathscr C_\mathfrak{o} + c$. By changing variables in the integration, we can conclude that this shift does not affect the quantities defined by \cref{eq:local_integral}. We then add zero rows and columns so that the operators are defined on the same spatial domain $\Lambda \cup \Lambda^\mathrm{ref}$. For example, for 
$\ell,k \in \Lambda \cup \Lambda^\mathrm{ref}$,
if $\ell \in \Lambda^\mathrm{ref} \setminus \Lambda$ or $k \in \Lambda^\mathrm{ref} \setminus \Lambda$, we redefine 
$\widetilde{\Ham}(u;\rho)_{\ell k}^{ab} \coloneqq 0$.
This only affects the spectrum by adding zero as an eigenvalue of finite multiplicity and so, because $0$ is not encircled by the contour $\mathscr C_\mathfrak{o}$, the value of \cref{eq:local_integral} is unchanged. For full details see \cite{ChenOrtnerThomas2019:locality}.

By replacing 
${\Ham}(u;\rho)$ by $\widetilde{\Ham}(u;\rho)$ 
in \cref{eq:local_integral} we obtain 
$\widetilde{O}_\ell(u;\rho) \coloneqq O_\ell(u;\rho) \chi_\Lambda(\ell)$.
In particular, if we write $\widetilde{F}(u;\rho)$ as a function of electronic densities defined on $\Lambda \cup \Lambda^\mathrm{ref}$, we find that the Jacobian of $\widetilde{F}(u;\rho)$ with respect to $\rho$, which we denote by
$\widetilde{\mathscr L}(u;\rho)$,
is obtained from 
${\mathscr L}(u;\rho)$ 
by inserting finitely many additional zero rows and columns. Therefore, 
$I - \widetilde{\mathscr L}(u;\rho)$ 
is invertible with 
$
\| \big(
        I-\widetilde{\mathscr L}(u;\rho)
    \big)^{-1}
\|_{\ell^2\to\ell^2} 
\leq \max\{
            1, 
            \mathsf{d}_{\mathscr L}^{-1} 
        \}
$.  

For the remainder of this paper, we consider the redefined quantities 
$\widetilde{\Ham}(u;\rho), \widetilde{\Ham}^\mathrm{ref}$
and 
$\widetilde{\mathscr L}(u;\rho), \widetilde{\mathscr L}^\mathrm{ref}$
and drop the tilde in the notation. 

\subsection{Proof of Theorem~\ref                               
    {thm:improved_locality}: Improved Locality Estimates}       
\label{sec:prooftheorem4}                                       

We now show that we can apply Lemma~\ref{lem:CT-improved} to the Hamiltonian and the stability operators and thus conclude the proof of Theorem~\ref{thm:improved_locality}.

\begin{lemma}[Perturbation of the Hamiltonian]
\label{lem:Ham_perturbation}
    Let $u \in \dot{\mathscr W}^{1,2}(\Lambda)$ satisfy $\asNonInter$ and $\rho$ be an associated electronic density satisfying \asSC~and \asFF.  
    %
    %
    Then, for all $\delta > 0$, there exist operators $\Ham^\delta = \Ham^\delta(u;\rho)$ and $\Ham^\mathrm{FR} = \Ham^\mathrm{FR}(u;\rho)$ such that
    \begin{align}\label{eq:Ham_pert}
        \Ham(u;\rho) = \Ham^\mathrm{ref} + \Ham^\mathrm{FR} + \Ham^\delta
    \end{align}
    where 
    $\|\Ham^\delta\|_\mathrm{F} \leq \delta$
    and there exists an $R>0$ such that 
    $[\Ham^\mathrm{FR}]_{\ell k}^{ab} = 0$ for all $(\ell,k) \not\in (\Lambda\cap B_R)^2$.
\end{lemma}
\begin{proof}
Applying \cite[Lemma~9]{ChenOrtnerThomas2019:locality}, we may conclude that \cref{eq:Ham_pert} holds for the linear Hamiltonian $\Ham^\mathrm{L}(u)$:
\[
    \Ham(u;\rho) = \Ham^\mathrm{ref} + P(u) + Q(u) + D(\rho)
\]
where 
$\|P(u)\|_\mathrm{F} \leq \delta$, 
there exists an $R>0$ such that 
$Q(u)_{\ell k}^{ab} = 0$ for all $(\ell,k) \not\in (\Lambda\cap B_R)^2$
and $D(\rho)$ is a diagonal operator with 
$D_{\ell\ell}^{aa} = v(\rho(\ell)) - v(\rho^\mathrm{ref}(\ell))$
for all $\ell \in \Lambda \setminus B_{R_\mathrm{ref}}$. After defining,
\begin{gather*}
    D^R(\rho)_{\ell\ell}^{aa}
    \coloneqq   \begin{cases}
                    D_{\ell\ell}^{aa}(\rho) 
                        &\textup{if } \ell \in \Lambda \cap B_R \\
                    0 &\textup{otherwise,}
                \end{cases}
    \quad \textup{we have} \\
    \lim_{R\to\infty} 
        \|D(\rho) - D^R(\rho)\|_{\ell^2 \to \ell^2} 
    = \limsup_{|\ell| \to \infty} 
        |v(\rho(\ell)) - v(\rho^\mathrm{ref}(\ell))| = 0.
\end{gather*}
That is, $D(\rho)$ may be approximated with appropriate finite rank operators.
\end{proof}
\begin{remark}
We remark here that if \asFF~is not satisfied then $D(\rho)$ from the proof of Lemma~\ref{lem:Ham_perturbation} is not compact and thus $\Ham(u;\rho) - \Ham^\mathrm{ref}$ is also not compact. This means that, as noted at the end of Remark~\ref{rem:FF}, the main techniques used in the proof of Theorem~\ref{thm:improved_locality} cannot be applied. 
\end{remark}
We now use Lemma~\ref{lem:Ham_perturbation} to show an analogous result for the stability operator.

\begin{lemma}[Perturbation of the Stability Operator]
\label{lem:stability_perturbation}
Let $u \in \dot{\mathscr W}^{1,2}(\Lambda)$ satisfy $\asNonInter$ and $\rho$ be an associated electronic density satisfying \asSC~and \asFF. 
 Then, for all $\delta > 0$, there exist operators $\mathscr L^\delta$ and $\mathscr L^\mathrm{FR}$ such that
\begin{align*}
    \mathscr L(u;\rho) 
        = \mathscr L^\mathrm{ref} + \mathscr L^\mathrm{FR} + \mathscr L^\delta.
\end{align*}
where 
$\|\mathscr L^\delta\|_\mathrm{F} \leq \delta$ 
and there exists an $R>0$ such that 
$[\mathscr L^\mathrm{FR}]_{\ell k} = 0$ for all $(\ell,k) \not\in (\Lambda\cap B_R)^2$.
\end{lemma}
\begin{proof}
    Using the notation from Lemma~\ref{lem:Ham_perturbation}, we may apply Lemma~\ref{lem:Woodbury} and obtain: for $z \in \mathscr C_f$, 
    \begin{align}
    \label{eq:resolvent_finite_rank_update}
    \begin{split}
        \mathscr R_z(u;\rho) &= 
        \big( 
            \Ham^\mathrm{ref} + \Ham^\mathrm{FR} - z 
        \big)^{-1} 
        + \big(
            \mathscr R_z(u;\rho) 
              - \big(
                    \Ham^\mathrm{ref} + \Ham^\mathrm{FR} - z 
                \big)^{-1}
        \big) \\
        &= \mathscr R^\mathrm{ref}_z 
        - \mathscr R^\mathrm{ref}_z 
            \big(
                I + \Ham^\mathrm{FR} \mathscr R^\mathrm{ref}_z
            \big)^{-1} 
            \Ham^\mathrm{FR} 
            \mathscr R^\mathrm{ref}_z
        - \mathscr R_z(u;\rho) 
            \Ham^\delta \big( 
                            \Ham^\mathrm{ref} + \Ham^\mathrm{FR} - z 
                        \big)^{-1}
        \\
        &\eqqcolon \mathscr R^\mathrm{ref}_z + P_z + Q_z.
    \end{split}
    \end{align}
    Therefore, using \cref{eq:L}, we have
    \begin{gather}
       \mathscr L(u;\rho)_{\ell k} - [\mathscr L^\mathrm{ref}]_{\ell k}
       = \frac{1}{2\pi i} \oint_{\mathscr{C}_f} f(z-\mu) 
           \sum_{ab} \big[ P_z + Q_z \big]_{\ell k}^{ab} 
            \big[ 
                2 \mathscr R_z^\mathrm{ref} + P_z + Q_z 
            \big]_{\ell k}^{ab} \mathrm{d}z\, v^\prime(\rho(k))
    \end{gather}
    and so, using the fact that 
    $\|2 \mathscr R_z^\mathrm{ref} + P_z + Q_z \|_\mathrm{max} < \infty$ 
    and arguing as in \cref{eq:resolvent_extra_bit}, we have
    \begin{align}\label{eq:stab_CT_estimate}
    \begin{split}
        \|\mathscr L(u;\rho) - \mathscr L^\mathrm{ref}\|_\mathrm{F}^2 
        &\leq C\sup_{z \in \mathscr C_f} \|P_z + Q_z\|_\mathrm{F}^2 
        \leq  C \sup_{z \in \mathscr C_f} (\|P_z\|_\mathrm{F}  + \|Q_z\|_\mathrm{F})^2 \\
        &\leq C\sum_{\ell,k \in \Lambda} e^{-\gamma_1 (|\ell| + |k|)}
        + C \sum_{\ell, \ell_1, \ell_2, k \in \Lambda} \sum_{a_1,a_2}
            e^{-\gamma_2 r_{\ell \ell_1}} 
            \big|
                [\Ham^\delta]^{a_1a_2}_{\ell_1 \ell_2}
            \big|^2 
            e^{-\gamma_3 r_{\ell_2 k}} \\
        &\leq C\bigg(
                    \sum_{\ell \in \Lambda} e^{-\gamma_1 |\ell|}
                \bigg)^2 
            + C \| \Ham^\delta \|^2_\mathrm{F} \leq C
    \end{split}
    \end{align}
    where 
        $\gamma_l \coloneqq 2\gamma_\mathrm{r}(\mathfrak{d}_l)$ 
    and $\gamma_\mathrm{r}$ is the constant \cref{eq:resolvent_bound} with appropriate choices of $\mathfrak{d}_l$ for $l=1,2,3$. Here, we have implicitly assumed that $\delta$ is sufficiently small such that for all $z \in \mathscr C_f$, we have 
    $\mathrm{dist}
        (z, \sigma(
                    \Ham^\mathrm{ref} + \Ham^\mathrm{FR}
                )
        ) \geq \mathfrak{d}_3 > 0$.
    Therefore, by applying \cref{eq:stab_CT_estimate}, for sufficiently large $R$, we can define
    \[
        [ \mathscr{L}^\mathrm{FR} ]_{\ell k} \coloneqq 
        \begin{cases}
            [\mathscr L(u;\rho) - \mathscr L^\mathrm{ref}]_{\ell k} 
                & \textup{if }\ell, k \in \Lambda \cap B_R \\
            0   & \textup{otherwise}
        \end{cases}
    \]
   and conclude 
   $
   \|\mathscr L(u;\rho) - \mathscr L^\mathrm{ref} - \mathscr{L}^\mathrm{FR}\|_\mathrm{F} \leq \delta
   $.
\end{proof}

For fixed $(u,\rho)$ and $\delta> 0$ sufficiently small, we fix operators $\Ham^\mathrm{FR}, \Ham^\delta$ and $\mathscr{L}^\mathrm{FR}, \mathscr{L}^\delta$ as in Lemmas~\ref{lem:Ham_perturbation} and \ref{lem:stability_perturbation} and apply Lemma~\ref{lem:CT-improved} to obtain: for $z \in \mathbb{C}$ with $\mathrm{dist}(z,\sigma(\Ham(u;\rho))) \geq \mathfrak{d} > 0$ and $\mathrm{dist}(z,\sigma(\Ham^\mathrm{ref})) - \delta \geq \mathfrak{d}^\mathrm{ref} > 0$ we have 
\begin{align}\label{eq:improved_resolvent_bound}\begin{split}
    \big| \mathscr{R}_z(u;\rho)_{\ell k}^{ab} \big| 
    &\leq C_{\ell k} e^{-\gamma_\mathrm{r}(\mathfrak{d}^\mathrm{ref}) |\ell - k| } \quad \textup{where} 
    \\
    {C}_{\ell k} \coloneqq  2(\mathfrak{d}^\mathrm{ref})^{-1} 
    &+ C (1 + |z|)^2 e^{-\gamma_\mathrm{r}(\mathfrak{d}^\mathrm{ref}) \left(|\ell| + |k| - |\ell - k|\right)} 
\end{split}\end{align}
and 
$
\gamma_\mathrm{r}(\mathfrak{d}^\mathrm{ref}) 
= c \min\{
            1,
            \mathfrak{d}^\mathrm{ref}
        \}
$
is the constant from \cref{eq:resolvent_bound}. We will apply \cref{eq:improved_resolvent_bound} with 
$z \in \mathscr C_f$ (with $\mathfrak{d} = \mathsf{d}_f$ 
and 
$
\mathfrak{d}^\mathrm{ref} 
    = \mathsf{d}_f^\mathrm{ref} \coloneqq 
    \min_{z \in \mathscr C_f} \mathrm{dist}( z,\sigma(\Ham^\mathrm{ref}) ) - \delta$)
and 
$z \in \mathscr C_\mathfrak{o}$
(with $\mathfrak{d} = \mathsf{d}_\mathfrak{o}$ and 
$
\mathfrak{d}^\mathrm{ref} 
= \mathsf{d}_\mathfrak{o}^\mathrm{ref} 
\coloneqq \min_{z \in \mathscr C_\mathfrak{o}} 
            \mathrm{dist}(z, \sigma(\Ham^\mathrm{ref})) - \delta
$). 

By \cref{eq:L}, we have 
$
    \big| \mathscr{L}(u;\rho)_{\ell k}\big| 
    \lesssim {C}_{\ell k}^2
    e^{-2\gamma_\mathrm{r}(\mathsf{d}_f^\mathrm{ref}) |\ell - k| } 
$
%
and so, applying Lemma~\ref{lem:CT-improved} again but now with $T$ replaced with $\mathscr{L}(u;\rho)$ (and with $z = 1$), we obtain, 
\begin{gather}
\label{eq:L_resolvent_bound_improved}
\begin{split}
    \big|
        \left[
            (I - \mathscr L(u;\rho))^{-1}
        \right]_{\ell k}^{ab} 
    \big|
    \leq \widetilde{C}_{\ell k} e^{-\gamma^\mathrm{ref}_\mathrm{s} r_{\ell k} } 
    \quad \textup{where} \\ 
    \widetilde{C}_{\ell k} 
    \coloneqq 2 (\mathsf{d}_{\mathscr L}^\mathrm{ref})^{-1} 
    + C e^{-\gamma^\mathrm{ref}_\mathrm{s} \left(|\ell| + |k| - |\ell - k|\right)}
\end{split}
\end{gather}
and 
$\gamma^\mathrm{ref}_\mathrm{s}
\coloneqq c_1 \gamma_\mathrm{r}(\mathsf{d}^\mathrm{ref}_f) 
\min\{1,C_{\ell k}^{-2}\gamma_\mathrm{r}(\mathsf{d}^\mathrm{ref}_f)^d \mathsf{d}^\mathrm{ref}_{\mathscr L}\}$
with 
$\mathsf{d}^\mathrm{ref}_{\mathscr L} 
    \coloneqq \mathrm{dist}(1,\sigma(\mathscr L^\mathrm{ref})) - \delta$. 
By expanding $\gamma_\mathrm{r}(\mathsf{d}_f^{\mathrm{ref}})$ in terms of $\mathsf{d}_f^\mathrm{ref}$, we obtain 
$
\gamma_\mathrm{s}^\mathrm{ref} 
    = c_2 \min\big\{    1,
                        \mathsf{d}^\mathrm{ref}_f,
                        C_{\ell k}^{-2} 
                            \min\{  1,
                                    (\mathsf{d}_f^\mathrm{ref})^{d+1}
                                \} \mathsf{d}^{\mathrm{ref}}_{\mathscr L}
                \big\}
$.
%

\begin{proof}[Proof of Theorem~\ref{thm:improved_locality}]
    The arguments in the proof of Theorem~\ref{thm:locality} can be applied with the resolvent estimates of \cref{eq:resolvent_bound} and \cref{eq:L_resolvent_bound} replaced with the corresponding improved estimates \cref{eq:improved_resolvent_bound} and \cref{eq:L_resolvent_bound_improved}. This means that the exponents 
    $\gamma_\mathrm{r}(\mathfrak{d})$ and $\gamma_\mathrm{s}$
    can be replaced with the improved exponents
    $\gamma_\mathrm{r}(\mathfrak{d}^\mathrm{ref})$
    and 
    $\gamma^\mathrm{ref}_\mathrm{s}$,
    respectively, and the pre-factors can be replaced with constants that depend on the atomic sites. These constants decay exponentially to the constants in the defect-free case as the subsystem moves away from the defect core together. This can be seen by noting that 
    $C_{\ell k} \to 2(\mathfrak{d}^\mathrm{ref})^{-1}$
    (where $C_{\ell k}$ is the constant from \cref{eq:improved_resolvent_bound}) and 
    $\widetilde{C}_{\ell k} \to 2(\mathsf{d}_{\mathscr L}^\mathrm{ref})^{-1}$
    (where $\widetilde{C}_{\ell k}$ is the constant from \cref{eq:L_resolvent_bound_improved}) as 
    $|\ell| + |k| - |\ell-k| \to \infty$ 
    with exponential rates. See \cite[(4.21)$-$(4.23)]{ChenOrtnerThomas2019:locality} for the analogous argument in the linear case that can be readily adapted to the setting we consider here.
\end{proof}

\appendix

\section{Band Structure of  \texorpdfstring
                                {$\Ham^\mathrm{ref}$}
                                {the Reference Hamiltonian} and 
                            \texorpdfstring
                                {$\mathscr L^\mathrm{ref}$}
                                {the Reference Stability Operator}}
\label{app:bands}

Recall that the unit cell 
$\Gamma \subset \Lambda^\mathrm{ref}$ 
is finite and satisfies 
$\Lambda^\mathrm{ref} = \bigcup_{\gamma \in \mathbb Z^d} (\Gamma + \mathsf{A}\gamma)$
and 
$\Gamma + \mathsf{A}\gamma$
pairwise disjoint for each $\gamma \in \mathbb Z^d$. Suppose 
$\Gamma^\star \subset \mathbb R^d$ 
is a bounded connected domain containing the origin and such that 
$\mathbb R^d 
    = \bigcup_{\eta \in \mathbb Z^d} 
        (\Gamma^\star + 2\pi \mathsf{A}^{-\mathrm{T}}\eta)$ 
and the 
$\Gamma^\star + 2\pi \mathsf{A}^{-\mathrm{T}}\eta$ 
are disjoint. Therefore, for each $\xi \in \mathbb R^d$, there exist unique $\xi_0 \in \Gamma^\star$ and $\eta \in \mathbb Z^d$ such that 
$\xi = \xi_0 + 2\pi \mathsf{A}^{-\mathrm{T}}\eta$ 
and, since 
$\mathsf{A}\gamma \cdot \mathsf{A}^{-\mathrm{T}}\eta = \gamma \cdot \eta$,
we have
$e^{-i\mathsf{A}\gamma\cdot\xi} = e^{-i\mathsf{A}\gamma \cdot \xi_0}$ 
for $\gamma \in \mathbb Z^d$.

Let us define the unitary operator $U \colon \ell^2(\Lambda^\mathrm{ref} \times \{1,\dots,\numorbitals\} ) \to L^2\left(\Gamma^\star;\ell^2(\Gamma \times \{1,\dots,\numorbitals\}) \right)$ by
\[
    (U\psi)_\xi(\ell;a) 
        =   \sum_{\gamma\in\mathbb Z^d} 
            \psi(\ell + \mathsf{A}\gamma;a)
            e^{-i(\ell + \mathsf{A}\gamma)\cdot\xi}.
\]
Here, $L^2\left(\Gamma^\star;\ell^2(\Gamma \times \{1,\dots,\numorbitals\})\right)$ is a Hilbert space with inner product
\[
    \Braket{\Psi,\Phi}_{L^2(\Gamma^\star;\ell^2)} 
        \coloneqq   \frac{1}{|\Gamma|} 
                    \int_{\Gamma^\star} 
                        \Braket{\Psi_\xi,\Phi_\xi}_\ell 
                    \mathrm{d}\xi
        = \frac{1}{|\Gamma|} 
            \sum_{\ell \in \Gamma} 
            \sum_{1\leq a \leq \numorbitals}
            \int_{\Gamma^\star} 
                \Psi_\xi(\ell;a)\overline{\Phi_\xi(\ell;a)} 
            \mathrm{d}\xi.
\]
A simple calculation reveals that  
$
    ( U\Ham^\mathrm{ref} \psi )_\xi 
    = \Ham^\mathrm{ref}_\xi ( U \psi )_\xi
$
where 
\[
    [ \Ham^\mathrm{ref}_\xi ]_{\ell k}^{ab} =
    \sum_{\gamma \in \mathbb Z^d} 
    h_{\ell k}^{ab}(\ell - k + \mathsf{A}\gamma) 
    e^{-i(\ell - k + \mathsf{A}\gamma)\cdot\xi} 
        + \delta_{\ell k}\delta_{ab} 
        \sum_{\gamma \in \mathbb Z^d}
        v(\rho^\mathrm{ref}(\ell)) 
        e^{-i\mathsf{A}\gamma\cdot\xi} .
\]
Letting $\lambda_n(\xi)$ be the ordered eigenvalues of $\Ham^\mathrm{ref}_\xi$ for $\xi \in \Gamma^\star$ and $n = 1,\dots,\numorbitals\cdot\#\Gamma$, we can use the fact that $U$ is unitary to conclude that,
\[
    \sigma(\Ham^\mathrm{ref}) 
        = \bigcup_n \bigcup_{\xi \in \Gamma^\star} \lambda_n(\xi).
\]
Since $\lambda_n \colon \Gamma^\star \to \mathbb R$ are continuous, we may conclude that $\sigma(\Ham^\mathrm{ref})$ is composed of finitely many spectral bands.

Similarly, 
\begin{gather*}
    \sigma(\mathscr{L}^\mathrm{ref}) 
    = 
    \bigcup_{\xi \in \Gamma^\star} \sigma(\mathscr{L}^\mathrm{ref}_\xi),
    \qquad \text{where} \\
    \big[\mathscr{L}^\mathrm{ref}_\xi\big]_{\ell k}
    = 
    \sum_{\gamma \in \mathbb Z^d} 
    \frac{1}{2\pi i} 
    \oint_{\mathscr C_f} 
        f(z-\mu) \sum_{a,b = 1}^{\numorbitals} 
        \left(
            [\mathscr R_z^\mathrm{ref}]_{\ell+\mathsf{A}\gamma,k}^{ab}
        \right)^2
    \mathrm{d}z\, v^\prime(\rho(k))
    e^{-i(\ell - k + \mathsf{A}\gamma)\cdot\xi}.
\end{gather*}

\bibliography{refs}
\bibliographystyle{siam}

\end{document}